\documentclass[a4paper,USenglish,cleveref, nameinlink, autoref, thm-restate, numberwithinsect]{lipics-v2021}
\usepackage{graphicx} 
\usepackage{amsmath}
\usepackage{tikz}
\usetikzlibrary{decorations.pathreplacing,arrows,math,shapes}
\usetikzlibrary{decorations.pathmorphing}
\usepackage{amsthm}
\usepackage{bm}
\theoremstyle{definition}
\newtheorem{problem}{Problem}

\usetikzlibrary{positioning}

\title{A Graph Width Perspective on Partially Ordered Hamiltonian Paths and Cycles II}
\subtitle{Vertex and Edge Deletion Numbers}
\titlerunning{A Graph Width Perspective on Partially Ordered Hamiltonian Paths and Cycles II}

\author{Jesse Beisegel}{Institute of Mathematics, Brandenburg University of Technology, Cottbus, Germany}{jesse.beisegel@b-tu.de}{https://orcid.org/0000-0002-8760-0169}{}
\author{Katharina Klost}{Institute of Computer Science, Freie Universität Berlin, Germany}{katharina.klost@fu-berlin.de}{https://orcid.org/0000-0002-9884-3297}{}
\author{Kristin Knorr}{Institute of Computer Science, Freie Universität Berlin, Germany}{kristin.knorr@fu-berlin.de}{https://orcid.org/0000-0003-4239-424X}{}
\author{Fabienne Ratajczak}{Institute of Mathematics, Brandenburg University of Technology, Cottbus, Germany}{fabienne.ratajczak@b-tu.de}{https://orcid.org/0000-0002-5823-1771}{}
\author{Robert Scheffler}{Institute of Mathematics, Brandenburg University of Technology, Cottbus, Germany}{robert.scheffler@b-tu.de}{https://orcid.org/0000-0001-6007-4202}{}

\authorrunning{J. Beisegel, K. Klost, K. Knorr, F. Ratajczak, and R. Scheffler}

\tikzstyle{vertex}=[draw, circle, fill=black, inner sep=1.5pt]

\renewcommand{\P}{\ensuremath{\mathsf{P}}}
\newcommand{\NP}{\ensuremath{\mathsf{NP}}}
\newcommand{\XP}{\ensuremath{\mathsf{XP}}{}}
\newcommand{\FPT}{\ensuremath{\mathsf{FPT}}}
\newcommand{\W}{\ensuremath{\mathsf{W[1]}}}
\newcommand{\MSO}{\ensuremath{\mathsf{MSO}}}
\newcommand{\cP}{\ensuremath{\mathcal{P}}}

\renewcommand{\O}{\ensuremath{\mathcal{O}}}
\newcommand{\G}{\ensuremath{\mathcal{G}}}
\newcommand{\T}{\ensuremath{\mathcal{T}}}

\newcommand{\N}{\ensuremath{\mathbb{N}}}
\newcommand{\R}{\ensuremath{\mathcal{R}}}
\renewcommand{\S}{\ensuremath{\mathcal{S}}}

\newcommand{\pip}{{\pi'}}

\renewcommand{\phi}{\varphi}

\newcommand{\paramfont}{\slshape}
\newcommand{\param}[1]{{\paramfont #1}}

\crefname{claim}{Claim}{Claims}
\Crefname{claim}{Claim}{Claims}

\crefname{observation}{Observation}{Observations}
\Crefname{observation}{Observation}{Observations}

\newcommand{\pef}[1]{(P\ref{#1})}



\nolinenumbers

\hideLIPIcs
\begin{document}

\maketitle

\begin{abstract}
We consider the problem of finding a Hamiltonian path or cycle with precedence constraints in the form of a partial order on the vertex set. We study the complexity for graph width parameters for which the ordinary problems \textsc{Hamiltonian Path} and \textsc{Hamiltonian Cycle} are in  \FPT. In particular, we focus on parameters that describe how many vertices and edges have to be deleted to become a member of a certain graph class. We show that the problems are \W-hard for such restricted cases as vertex distance to path and vertex distance to clique. We complement these results by showing that the problems can be solved in \XP{} time for vertex distance to outerplanar and vertex distance to block. Furthermore, we present some \FPT{} algorithms, e.g., for edge distance to block. Additionally, we prove para-\NP-hardness when considered with the edge clique cover number.
\end{abstract}

\section{Introduction}

\subparagraph*{Hamiltonian Paths and Cycles with Precedence} 

For some applications of the well-known \textsc{Traveling Salesman Problem} (TSP) it is necessary to add additional precedence constraints to the vertices which ensure that some vertices are visited before others in a tour. Examples are \emph{Pick-up and Delivery}  Problems~\cite{parragh2008survey,parragh2008survey2} or the \emph{Dial-a-Ride} problem~\cite{psaraftis1980dynamic}, where goods or people have to be picked up before they can be brought to their destination.

The generalization of TSP with this type of constraint is known as \textsc{Traveling Salesman Problem with Precedence Constraints} (TSP-PC) and has been studied, e.g., in \cite{ahmed2001travelling,bianco1994exact}. If we are looking for a path instead of a cycle, this is known as the \textsc{Sequential Ordering Problem} (SOP) or the \textsc{Minimum Setup Scheduling Problem} and has been studied, e.g., in~\cite{ascheuer1993cutting,colbourn1985minimizing,escudero1988inexact,escudero1988implementation}. 

As TSP is already \NP-complete without the precedence constraints, research on these problems has mainly focused on the practical applications and solution methods are restricted to heuristic algorithms and integer programming.
When moving to the structurally different Hamiltonian path (cycle) problem, only very restricted precedence constraints have been studied, e.g., where one or both endpoints of the Hamiltonian path are fixed. For these problems, polynomial-time algorithms have been presented for several graph classes including (proper) interval graphs~\cite{asdre2010fixed,asdre2010polynomial,li2017linear,mertzios2010optimal} and rectangular grid graphs~\cite{itai1982hamilton}.

In~\cite{beisegel2024computing}, Beisegel et al. introduce the \textsc{Partially Ordered Hamiltonian Path Problem (POHPP)}, where we are given a graph together with a partial order and search for a Hamiltonian path that is a linear extension of the partial order.
Additionally, they introduced the edge-weighted variant \textsc{Minimum Partially Ordered Hamiltonian Path Problem} (MinPOHPP). The authors show that POHPP is already \NP-hard for complete bipartite graphs and complete split graphs -- graph classes where \textsc{Hamiltonian Path} is trivial.
They also show that  POHPP is \W-hard when parameterized by the \param{width} of the partial order, i.e., the largest number of pairwise incomparable elements. Furthermore, they show that the \XP{} algorithm for that parameterization presented in the 1980s by Coulbourn and Pulleyblank~\cite{colbourn1985minimizing} is asymptotically optimal -- assuming the Exponential Time Hypothesis (ETH). They improve the algorithm to \FPT{} time if the problem is parameterized by the partial order's \param{distance to linear order}. Finally, the authors present a polynomial-time algorithm for MinPOHPP on outerplanar graphs.

\subparagraph*{Graph Width Parameters and Hamiltonicity} 

In~\cite{beisegel2024computing}, Beisegel et al. imply that their result for outerplanar graphs might be extended to series parallel graphs, otherwise known as the graphs of treewidth two. Furthermore, they suggest further research on the complexity of (Min)POHPP for graphs of bounded treewidth as well as other graph width parameters. Just as for specific graph classes, the prerequisite to solving (Min)POHPP for some graph parameter is the tractability of the base problems \textsc{Hamiltonian Path} and \textsc{Hamiltonian Cycle}. These belong to the most famous \NP-complete graph problems and, hence, have been studied for a wide range of graph classes and graph width parameters.\footnote{In general, \textsc{Hamiltonian Path} seems to lead a shadowy existence since many positive and negative algorithmic results are only given for its more popular sibling \textsc{Hamiltonian Cycle}.}

One of the first results for width parameters were polynomial-time algorithms for \textsc{Hamiltonian Cycle} and TSP on graphs of bounded \param{bandwidth}~\cite{lawler1985traveling,monien1980bounding}. Since both \textsc{Hamiltonian Cycle} and \textsc{Hamiltonian Path} are expressible in $\MSO_2$, Courcelle's theorem~\cite{courcelle1990monadic,courcelle1992monadic} implies \FPT{} algorithms for these problems when parameterized by \param{treewidth}. The running time bounds depending on the \param{treewidth}~$k$ given by Courcelle's theorem are quite bad. However, there are also \FPT{} algorithms with single exponential dependency on $k$~\cite{bodlaender2015deterministic,cygan2018fast,cygan2022solving}. Probably, these results cannot be transferred to \param{cliquewidth}: In contrast to $\MSO_2$, both \textsc{Hamiltonian Cycle} and \textsc{Hamiltonian Path} are not expressible in $\MSO_1$~\cite[Corollary~5.3.5]{ebbinghaus1995finite}. In fact, \textsc{Hamiltonian Cycle} is \W-hard when parameterized by \param{cliquewidth}~\cite{fomin2010intractability} and -- unless $\FPT = \W$ -- it is not solvable in \FPT{} time. Furthermore, it was shown that -- unless the ETH fails -- the problem cannot be solved in $f(k) n^{o(k)}$ time where $k$ is the \param{cliquewidth} and $n$ is the number of vertices~\cite{fomin2018clique-width}. This lower bound is matched by an $f(k) n^{\O(k)}$ algorithm given by Bergougnoux et al.~\cite{bergougnoux2020optimal}. \XP{} algorithms for \param{cliquewidth} with worse exponents have been presented earlier by Wanke~\cite{wanke1994k-nlc} (for the equivalent NLC-width) and by Espelage et al.~\cite{espelage2001solve}. 

As \FPT{} algorithms for \textsc{Hamiltonian Path} and \textsc{Hamiltonian Cycle} parameterized by \param{cliquewidth} are highly unlikely, the research has focused on upper bounds of \param{cliquewidth}. \FPT{} algorithms were presented for \param{neighborhood diversity}~\cite{lampis2012algorithmic}, \param{distance to cluster}~\cite{doucha2012cluster,jansen2013power}, \param{modular width}~\cite{gajarsky2013parameterized} and \param{split matching width}~\cite{saether2016between,saether2017solving}. Besides this, parameterized algorithms have also been given for graph width parameters incomparable to \param{cliquewidth} and \param{treewidth}. Examples are \FPT{} algorithms for \param{distance to proper interval graphs}~\cite{golovach2020graph} and for two parameters that describe the distance to Dirac's condition on the existence of Hamiltonian paths~\cite{jansen2019hamiltonicity}. Most recently, \FPT{} algorithms for the \param{independence number} have been presented~\cite{fomin2025path}. 

The parameters \param{treewidth}, \param{pathwidth}, and \param{bandwidth} in the context of (Min)POHPP have been studied in the companion paper~\cite{beisegel2025graphi}, where we showed that POHPP is \NP-hard for small constants. Furthermore, it follows directly from the NP-completeness of POHPP on complete bipartite graphs that it is also \NP-complete on graphs of \param{cliquewidth}~$2$, \param{neighborhood diversity}~$2$, \param{modular width}~$2$ and \param{split matching width}~$1$. As implied before, the next likely candidates to yield \FPT{} or \XP{} algorithms are the parameters that consider the smallest number of vertices or edges that have to be removed such that the resulting graph is in some given graph class. We refer to these parameters as \param{(vertex) distance to $\G$} and \param{edge distance to $\G$} where $\G$ is the respective graph class.\footnote{Alternative terms for these parameters are \param{$\G$ vertex deletion number} and \param{$\G$ edge deletion number}.} 
Here, we will also consider a special variant of these distance parameters that is motivated by the parameter \param{twin-cover number} introduced by Ganian~\cite{ganian2011twin-cover}. We say a graph has \param{distance to $\G$ module(s)}\footnote{If the class $\G$ contains only connected graphs, then we use \param{distance to $\G$ module} otherwise we use \param{distance to $\G$ modules}.}~$k$ if there is a set $W \subseteq V(G)$ of $k$ vertices such that $G - W$ is in $\G$ and every component of $G - W$ forms a module in $G$. Using this terminology, the \param {twin-cover number} is equal to the \param{distance to cluster modules}. 

\subparagraph*{Our Contribution} 

\begin{figure}
\centering
\resizebox{\textwidth}{!}{
\begin{tikzpicture}[xscale=2.5,yscale=1]
  \tikzset{np/.style={draw, trapezium, trapezium angle=77.5, fill=red!50!white}}
  \tikzset{fpt/.style={draw,rectangle, fill=green!30!lightgray,rounded corners,
  }}
  \tikzset{xpw/.style={draw,rectangle,fill=orange!50!white}}
  \tikzset{xp/.style={draw,rectangle,fill=yellow!20!white}}
  \tikzset{w/.style={draw,rectangle,fill=yellow!70!white}}
  \tikzset{fptnp/.style={draw,rectangle, fill=green!30!lightgray,  draw=red!50!white, line width=2pt,
  rounded corners
  }}
  \tikzset{minnp/.style={draw=red!50!white, line width=2pt}}
  \tikzset{xpnp/.style={draw,rectangle,fill=yellow!20!white, draw=red!50!white, line width=2pt, 
  }}
  \tikzset{wnp/.style={draw,rectangle,fill=yellow!70!white,  draw=red!50!white, line width=2pt, 
  }} 
  \tikzset{xpwnp/.style={draw,rectangle,fill=orange!50!white,  draw=red!50!white, line width=2pt,
  }}
  \tikzset{open/.style={draw,rectangle}}
  \scriptsize
  \paramfont
  \node[fpt] (vc) at (-0.95,0) {\mathstrut Vertex Cover};
  \node[xpw, align=center] (d2path) at (-0.36,1) {\mathstrut Distance to \\ Path};
  \node[xpwnp, align=center] (d2clique) at (2.5,2) {\mathstrut Distance to \\ Clique};
  \node[xpwnp, align=center] (d2block) at (0.3,3) {\mathstrut Distance to \\ Block};
  \node[fptnp, align=center] (ed2block) at (1.1,1.3) {\mathstrut Edge Distance \\ to Block};
  \node[xpwnp, align=center] (tc) at (2,1.1) {\mathstrut Twin-Cover};
  \node[np, align=center] (ecc) at (3.5,2) {\mathstrut Edge \\ Clique Cover};
  \node[np, align=center] (nd) at (2.5,3) {\mathstrut Neighborhood \\ Diversity~{\normalfont \cite{beisegel2024computing}}};
  \node[np, align=center] (cc) at (3.5,3.1) {\mathstrut Clique Cover};
  \node[np, align=center] (is) at (3.5,4) {\mathstrut Independent Set};
  \node[fpt] (td) at (-2.1,0.5) {\mathstrut Treedepth};
  \node[fpt, align=center] (fes) at (0.275,1) {\mathstrut Feedback \\ Edge Set};
  \node[np] (bw) at (-1.9,1.65) {\mathstrut Bandwidth~{\normalfont \cite{beisegel2025graphi}}};
  \node[np] (pw) at (-2,3) {\mathstrut Pathwidth~{\normalfont \cite{beisegel2025graphi}}};
  \node[np] (tw) at (-1.25,3.75) {\mathstrut Treewidth~{\normalfont \cite{beisegel2025graphi}}};
  \node[xpw, align=center] (fvs) at (-0.4,2.3) {\mathstrut Feedback \\ Vertex Set};
  \node[xpw, align=center] (d2op) at (-1.25,2.95) {\mathstrut Distance to \\ Outerplanar};
  \node[np, align=center] (d2pi) at (1.5,3) {\mathstrut Distance to \\ Proper Interval~{\normalfont \cite{beisegel2025graphi}}};
  \node[np, align=center] (cw) at (1.25,4.25) {\mathstrut Cliquewidth~{\normalfont \cite{beisegel2024computing}}};
  \node[fptnp, align=center] (d2cm) at (2.5,0) {\mathstrut Distance to \\ Clique Module};
  \node[xpw, align=center] (d2lfm) at (-1.1,1.1) {\mathstrut Distance to \\ Linear Forest \\ Modules};
   \draw[thick,-stealth'] (bw) -- (pw);
   \draw[thick,-stealth'] (pw) -- (tw);
   \draw[thick,-stealth'] (vc.0) -- (tc);
   \draw[thick,-stealth'] (vc.180) -- (td.345);
   \draw[thick,-stealth'] (td.160) -- (pw.-166);
   \draw[thick,-stealth'] (d2clique) -- (cc);
   \draw[thick,-stealth'] (d2clique) -- (ecc);
   \draw[thick,-stealth'] (d2op) -- (tw);
   \draw[thick,-stealth'] (ecc) -- (nd);
   \draw[thick,-stealth'] (cc) -- (is);
   \draw[thick,-stealth'] (d2block) -- (cw);
   \draw[thick,-stealth'] (tw) -- (cw);
   \draw[thick,-stealth'] (nd) -- (cw);
   \draw[thick,-stealth'] (fes) -- (ed2block);
   \draw[thick,-stealth'] (ed2block.166) -- (d2block);
   \draw[thick,-stealth'] (d2cm) -- (d2clique);
   \draw[thick,-stealth'] (d2cm) -- (tc);
   \draw[thick,-stealth'] (d2path) -- (d2pi);
   \draw[thick,dashed,-stealth'] (td.355) -- (vc.173);
   \draw[thick,dashed,-stealth'] (vc.0) -- (d2cm.170);
   \draw[thick,dashed,-stealth'] (vc.170) -- (bw.200);
   \draw[thick,dashed,-stealth'] (vc.0) -- (fes.215);
   \draw[thick,dashed,-stealth'] (ecc) -- (cc);

   \draw[thick, -stealth'] (d2path.105) -- (fvs);
   \draw[thick, -stealth'] (fvs) -- (d2block);
   \draw[thick, -stealth'] (d2lfm) -- (fvs);
   \draw[thick, -stealth'] (d2lfm) -- (d2pi);
   \draw[thick, -stealth'] (fvs) -- (d2block);
   \draw[thick, -stealth'] (tc) -- (d2block.-30);
   \draw[thick, dashed, -stealth'] (tc.0) -- (ecc);
   \draw[thick, -stealth'] (d2clique.200) -- (d2block.-20);
   \draw[thick, -stealth'] (d2lfm) -- (pw);
   \draw[thick, -stealth'] (tc) -- (d2pi);
   \draw[thick, -stealth'] (d2clique) -- (d2pi);
   \draw[thick, -stealth'] (fes) -- (fvs);
   \draw[thick, -stealth'] (d2path.105) -- (pw);
   \draw[thick, -stealth'] (fvs) -- (d2op);
   \draw[thick, dashed, -stealth'] (vc) -- (d2path);
   \draw[thick, -stealth'] (vc) -- (d2lfm);
   \draw[thick, -stealth'] (1.5,0.89) -- (nd.200);

   \node[draw=black] at (0.75,-1) {\begin{tikzpicture}
         \normalfont
       \node (Legend) at (-2,-1) {Legend:};
       \node[fpt] (FPT) [right=0.5cm of Legend] {\FPT};
       \node[np] (NP) [right=0.5cm of FPT] {para-\NP-hard};
       \node[xpw] (XPW) [right=0.5cm of NP] {\XP{} and \W-hard};
       \node[minnp] (MINNP) [right=0.5cm of XPW] {MinPOHPP para-\NP-hard};
   \end{tikzpicture}};
   
 \end{tikzpicture}
}
\caption{Diagram illustrating the complexity results for (Min)POHPP for different graph width parameters. A directed solid edge from parameter $P$ to parameter $Q$ means that a bounded value of $P$ implies a bounded value for $Q$. A directed dashed edge implies that this relation does not hold in general but for traceable graphs, i.e., graphs having a Hamiltonian path. If a directed solid path from $P$ to $Q$ is missing, then parameter $Q$ is unbounded for the graphs of bounded $P$. The same holds for the traceable graphs if there is also no path using dashed edges.}\label{fig:parameters}
\end{figure}
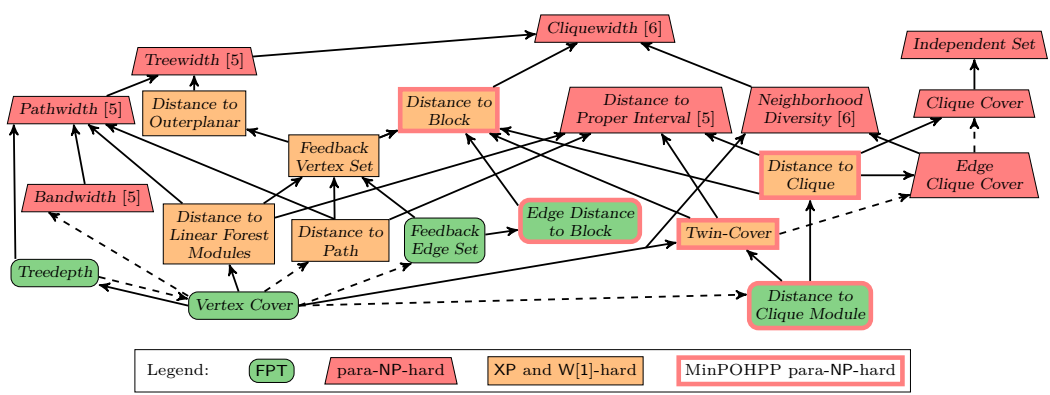

We study the complexity of (Min)POHPP and its cycle variant for graph width parameters where \textsc{Hamiltonian Path} and \textsc{Hamiltonian Cycle} can be solved in \FPT{} time (see \cref{fig:parameters} for an overview). We will show that the computational complexity of the path and cycle problems are essentially equivalent (see \cref{thm:path2cycle,thm:cycle2path}) and therefore all proofs will focus on the path version.

\Cref{sec:d2p} is dedicated to sparse parameters, i.e., graph parameters that are unbounded for complete graphs. We show that POHPP is \W-hard for \param{distance to path} and \param{distance to linear forest modules}. We present simple arguments why MinPOHPP can be solved in \FPT{} time for \param{feedback edge set number} and \param{treedepth}. Furthermore, we present an \XP{} algorithm for \param{distance to outerplanar}. We also give an \FPT{} algorithm for plane graphs parameterized by the number of vertices that lie in the interior of the outer face.

In \cref{sec:d2c}, we consider graph width parameters that are bounded for cliques. We show that POHPP is \NP-complete for graphs of \param{edge clique cover number}~3 and \param{clique cover number}~2. Furthermore, we prove \W-hardness for \param{distance to clique} and \param{distance to cluster modules} aka \param{twin cover number}. On the positive side, we present an \XP{} algorithm for POHPP when parameterized by \param{distance to block} and \FPT{} algorithms when parameterized by \param{edge distance to block} or \param{distance to clique module}.

\section{Preliminaries}

\subparagraph{General Notation and Partial Orders} A \emph{partial order} $ \pi $ on a set $X$ is a reflexive, antisymmetric and transitive relation on $X$. The tuple $(X, \pi)$ is then called a \emph{partially ordered set}. We also denote $(x,y) \in \pi$ by $x \prec_\pi y$ if $x \neq y$. If it is clear which partial order is meant, then we sometimes omit the index. A \emph{minimal element} of a partial order $\pi$ on $X$ is an element $x \in X$ for which there is no element $y \in X$ with $y \prec_\pi x$. The \emph{reflexive and transitive closure} of a relation $\R$ is the smallest relation $\R'$ such that $\R \subseteq \R'$ and $\R'$ is reflexive and transitive.

A \emph{linear ordering} of a finite set $X$ is a bijection $\sigma: X \rightarrow \{1,2,\dots,|X|\}$. We will often refer to linear orderings simply as orderings. Furthermore, we will denote an ordering by a tuple $(x_1, \ldots, x_n)$ which means that $\sigma(x_i) = i$. Given two elements $x$ and $y$ in $X$, we say that $x$ is \emph{to the left} (resp. \emph{to the right}) of $y$ if $\sigma(x)<\sigma(y)$ (resp. $\sigma(x)>\sigma(y)$) and we denote this by $x \prec_{\sigma}y$ (resp.  $x \succ_{\sigma}y$). A \emph{linear extension} of a partial order $\pi$ is a linear ordering $\sigma$ of $X$ that fulfills all conditions of $\pi$, i.e., if $x \prec_\pi y$, then $x \prec_\sigma y$.

For $n \in \N$, the notation $[n]$ refers to the set $\{i \in \N \mid 1 \leq i \leq n\}\).

\subparagraph{Graphs and Graph Classes}
All the graphs considered here are finite. For the standard notation of graphs we refer to the book of Diestel~\cite{diestel}.

For a subgraph $H$ of $G$, we say that two vertices $v,w \in V(H)$ are \emph{$H$-neighbors} if $vw \in E(H)$. A vertex $v$ of a connected graph $G$ is a \emph{cut vertex} if $G - v$ is not connected. If $G$ does not contain a cut vertex, then $G$ is \emph{2-connected}. A \emph{block} of a graph is an inclusion-maximal 2-connected induced subgraph. The \emph{block-cut tree} $\mathcal{T}$ of $G$ is the bipartite graph that contains a vertex for every cut vertex of $G$ and a vertex for every block of $G$ and the vertex of block $B$ is adjacent to the vertex of a cut vertex $v$ in $\mathcal{T}$ if $B$ contains $v$.

A graph is a \emph{linear forest} if all its connected components are paths. A graph is a \emph{cluster graph} if all its components are cliques. A graph is a \emph{block graph} if its blocks are cliques. 

A graph is \emph{planar} if it has a crossing-free embedding in the plane, and together with this embedding it is called a \emph{plane graph}. For a plane graph $G$ we call the regions of $\mathbb{R}^2 \setminus G$ the \emph{faces} of $G$. Every plane graph has exactly one unbounded face which is called the \emph{outer face}. Note that we will use the face and the closed walk of $G$ that forms the border of that face interchangeably. A graph is called \emph{outerplanar} if it has a crossing-free embedding such that all of the vertices belong to the outer face and such an embedding is also called \emph{outerplanar}.

\subparagraph{Graph Width Parameters}
A \emph{treedepth decomposition} of a graph consists of a rooted forest $F$ where every vertex of $G$ is mapped to a vertex in $F$ such that two vertices that are adjacent in $G$ have to form ancestor and descendant in $F$. The \param{treedepth} of a graph $G$ is the minimal height of a treedepth decomposition.

An \emph{(edge) clique cover} of a graph $G$ is a partition of the vertex set of $G$ (the edge set of $G$) into cliques. The \param{(edge) clique cover number} of $G$ is the minimal size of an (edge) clique cover of $G$.

Let $\G$ be a graph class. The \param{distance to $\G$} of a graph $G$ is the minimal number of vertices that need to be removed from $G$ so that the resulting graph belongs to $\G$. Similarly, the \param{edge distance to $\G$} of a graph $G$ is the minimal number of edges that need to be removed from $G$ so that the resulting graph belongs to $\G$. A non-trivial \emph{module} of a graph is a strict subset $M \subsetneq V(G)$ such that for all vertices $u, v \in M$ it holds that every neighbor of $u$ outside of $M$ is also a neighbor of $v$. The \param{distance to $\G$ module(s)} is the minimal number of vertices that need to be removed from $G$ so that the resulting graph is in $\G$ and every component of the resulting graph is a non-trivial module in $G$. Several of the distance parameters have got their own name. So \param{vertex cover number} is equivalent to \param{distance to edgeless}, \param{feedback vertex (edge) set number} is equivalent to \param{(edge) distance to tree} and \param{twin-cover number}~\cite{ganian2011twin-cover} is equivalent to \param{distance to cluster modules}.

\subparagraph{Hamiltonian Paths and Cycles}

A \emph{Hamiltonian path (cycle)} of a graph $G$ is a path (cycle) that contains all the vertices of $G$. A graph is \emph{traceable} if it has a Hamiltonian path and \emph{Hamiltonian} if it has a Hamiltonian cycle. Here, we only consider \emph{ordered} Hamiltonian paths, i.e., one of the two possible orderings of the path is fixed. Given a partial order $\pi$ on a graph's vertex set, an ordered Hamiltonian path is a \emph{$\pi$-extending Hamiltonian path} if its order is a linear extension of $\pi$. Using this definition we can generalize the classical \textsc{Hamiltonian Path} using precedence constraints.

\begin{problem}{\textsc{Partially Ordered Hamiltonian Path Problem} (POHPP)}
\begin{description}
\item[\textbf{Instance:}] A graph $G$, a partial order $\pi$ on the vertex set of $G$.
\item[\textbf{Question:}]
Is there an ordered Hamiltonian path $(v_1, \dots, v_n)$ in $G$ such that for all $i, j \in \{1,\dots,n\}$ it holds that if $(v_i,v_j) \in \pi$, then $i \leq j$?
 \end{description}
\end{problem}

A similar definition can be derived for \textsc{Hamiltonian Cycle}.

\begin{problem}{\textsc{Partially Ordered Hamiltonian Cycle Problem} (POHCP)}
\begin{description}
\item[\textbf{Instance:}] A graph $G$, a partial order $\pi$ on the vertex set of $G$.
\item[\textbf{Question:}]
Is there an ordered Hamiltonian path $(v_1, \dots, v_n)$ in $G$ such that $v_1$ and $v_n$ are adjacent and for all $i, j \in \{1,\dots,n\}$ it holds that if $(v_i,v_j) \in \pi$, then $i \leq j$?
 \end{description}
\end{problem}

The problems \textsc{Hamiltonian Path} and \textsc{Hamiltonian Cycle} without precedence constraints are independent from each other, i.e., there exist graph classes where one of the problems is trivial while the other is \NP-hard. It is easy to show that \textsc{Hamiltonian Path} is \NP-complete on the graphs that are not Hamiltonian.\footnote{Take a graph $G$ and add a universal vertex $u$ and a leaf that is only adjacent to $u$. The resulting graph is not Hamiltonian and is traceable if and only if $G$ is traceable.} In contrast, \textsc{Hamiltonian Cycle} is \NP-complete on traceable graphs~\cite[Lemma~21.18]{korte2018combinatorial}. Using the same arguments as in~\cite{beisegel2025graphi} for graph classes, we can show that for POHPP and POHCP the status of being in \FPT{} or \XP{} is equivalent for many graph width parameters.

To this end, let $\mathbb{G}$ be the family of all graphs. Let $\xi : \mathbb{G} \to \N$ be a graph width parameter. A \emph{graph operation} $\Omega$ is a function $\mathbb{G} \to \mathbb{G}$. We say that $\xi$ is \emph{closed under graph operation $\Omega$} if there is a computable function $h : \N \to \N$ such that for all graphs $G$, it holds that $\xi(\Omega(G)) \leq h(\xi(G))$.

\begin{observation}[see Observation~2.4 in~\cite{beisegel2025graphi}]\label{thm:path2cycle}
    Let $\xi$ be a graph width parameter that is closed under the addition of a universal vertex. Then there is a polynomial-time \FPT-reduction from POHPP parameterized by $\xi$ to POHCP parameterized by~$\xi$.
\end{observation}

For the converse direction, we are not able to give a (many-one) \FPT-reduction. Nevertheless, we can give an \FPT-Turing reduction.

\begin{observation}[see Observation~2.3 in~\cite{beisegel2025graphi}]\label{thm:cycle2path}
    Let $\xi$ be a graph width parameter. If we can solve (Min)POHPP in time $f(\xi(G)) \cdot n^{g(\xi(G))}$ on a graph $G$, then we can solve (Min)POHCP in time $f(\xi(G)) \cdot n^{g(\xi(G)) + 2}$ on $G$.
\end{observation}

Note that \cref{thm:path2cycle} does also hold for the classical problems \textsc{Hamiltonian Path} and \textsc{Hamiltonian Cycle}. In contrast, \cref{thm:cycle2path} does not hold for the classical problems since \textsc{Hamiltonian Cycle} is \NP-complete on traceable graphs~\cite[Lemma~21.18]{korte2018combinatorial}.

These two observations allow us to restrict our proofs to the path case. Whenever we prove \W-hardness for POHPP, the respective width parameter is closed under the addition of a universal vertex. So \cref{thm:path2cycle} implies \W-hardness also for the cycle case. If we present some \FPT{} or \XP{} algorithm for (Min)POHPP, then \cref{thm:cycle2path} also implies an \FPT{} or \XP{} algorithm for (Min)POHCP.

\subparagraph{Complexity} We will use the following three problems in reductions.
\begin{problem}{\textsc{Multicolored Clique Problem} (MCP)}
\begin{description}
\item[\textbf{Instance:}] A graph $G$ with a proper coloring by $k$ colors.
\item[\textbf{Question:}]
Is there a clique $C$ in $G$ such that $C$ contains exactly one vertex of each color?
 \end{description}
\end{problem}

The MCP was shown to be \W-hard by Pietrzak~\cite{pietrzak2003parameterized} and independently by Fellows et al.~\cite{fellows2009parameterized}. Note that this also holds if all color classes have the same size.

For non-empty disjoint sets $A$ and $B$ with $|A| = |B|$ and a partial order $\pi$ on $A \cup B$, we say that a linear extension $(x_1, \dots, x_n)$ of $\pi$ is an \emph{alternating} linear
extension if $x_i \in A$ if and only if $i$ is odd. 

\begin{problem}{\textsc{Alternating Linear Extension Problem}}
\begin{description}
\item[\textbf{Instance:}] Two non-empty disjoint sets $A$ and $B$ with $|A| = |B|$, partial order $\pi$ on $A \cup B$.
\item[\textbf{Question:}] Is there an alternating linear extension of $\pi$?
\end{description}
\end{problem}

Beisegel et al.~\cite{beisegel2024computing} showed that the \textsc{Alternating Linear Extension Problem} is \NP-complete even if the partial order is \emph{oriented from $A$ to $B$}, i.e., for all constraints $a \prec_\pi b$ it holds that $a \in A$ and $b \in B$.

\section{Sparse Width Parameters}\label{sec:d2p}
We start with parameters that are unbounded for cliques. As is shown in~\cite{beisegel2025graphi}, both POHPP and POHCP are para-\NP-hard when parameterized by \param{treewidth}. Therefore, we focus on parameters that form upper bounds of \param{treewidth}.

\subsection{Hardness}

We start by showing that even for very restricted \param{distance to \G} parameters, there is no hope for \FPT{} algorithms.

\begin{theorem}\label{thm:w1-d2p}
     POHPP and POHCP are \W-hard when parameterized by either \param{distance to path} or \param{distance to linear forest modules}.
\end{theorem}

\begin{proof}

\begin{figure}
    \centering
    \resizebox{\textwidth}{!}{
    \begin{tikzpicture}[xscale=1, clique/.style={draw, rounded corners}, multiedge/.style={line width=1.75}]
\footnotesize

        \node[clique, label=-45:$X^1$] (X1) at (0,0) {
        \begin{tikzpicture}
        \node[vertex, label=90:$x^1_1$] (1x11) at (0,0) {};
        \node[vertex, label=90:$\hat{x}^1_1$] (1hx11) at (0.5,0) {};
        \node[vertex, label=90:$x^1_2$] (1x12) at (1,0) {};
        \node[vertex, label=90:$\hat{x}^1_q$] (1x1q) at (1.75,0) {};
        \draw (1x11) -- (1hx11) -- (1x12) -- (1x1q);
        \node[fill=white, rectangle, rounded corners=false, inner sep=1pt] at (1.375,0) {\tiny $\dots$};
        \end{tikzpicture}
        };

        \node[clique, label=-45:$X^2$] (X2) [right=0.25cm of X1] {
        \begin{tikzpicture}
        \node[vertex, label=90:$x^2_1$] (2x11) at (0,0) {};
        \node[vertex, label=90:$\hat{x}^2_1$] (2hx11) at (0.5,0) {};
        \node[vertex, label=90:$x^2_2$] (2x12) at (1,0) {};
        \node[vertex, label=90:$\hat{x}^2_q$] (2x1q) at (1.75,0) {};
        \draw (2x11) -- (2hx11) -- (2x12) -- (2x1q);
        \node[fill=white, rectangle, rounded corners=false, inner sep=1pt] at (1.25,0) {\tiny $\dots$};
        \end{tikzpicture}
        };

        \node [right=0.1cm of X2] {$\dots$};

        \node[clique, label=-45:$X^k$] (Xk) [right=0.75cm of X2] {
        \begin{tikzpicture}
        \node[vertex, label=90:$x^k_1$] (3x11) at (0,0) {};
        \node[vertex, label=90:$\hat{x}^k_1$] (3hx11) at (0.5,0) {};
        \node[vertex, label=90:$x^k_2$] (3x12) at (1,0) {};
        \node[vertex, label=90:$\hat{x}^k_q$] (3x1q) at (1.75,0) {};
        \draw (3x11) -- (3hx11) -- (3x12) -- (3x1q);
        \node[fill=white, rectangle, rounded corners=false, inner sep=1pt] at (1.25,0) {\tiny $\dots$};
        \end{tikzpicture}
        };

        \node[vertex, label=90:$s^1$] (s1) [above left=0.5cm and -0.3cm of X1] {};
        \node[vertex, label=90:$\hat{t}^1$] (t1) [above=0.5cm of X1] {};
        \node[vertex, label=180:$\hat{s}^1$] (hs1) [below=0.5cm of X1] {};
        \node[vertex, label=90:$s^2$] (s2) [above right=0.25cm and 0.125cm of X1] {};
        \node[vertex, label=90:$\hat{s}^2$] (hs2) [above right=1cm and 0.125cm of X1] {};
        \node[vertex, label=90:$\hat{t}^2$] (t2) [above=0.5cm of X2] {};
        \node[vertex, label=90:$\hat{t}^k$] (tk) [above=0.5cm of Xk] {};
    
        \draw[multiedge] (t1) -- (X1);
        \draw[multiedge] (t2) -- (X2);
        \draw[multiedge] (tk) -- (Xk);
        \draw[multiedge] (hs1) -- (X1);
        \draw[multiedge] (s1) -- (X1);
        \draw[multiedge] (s2) -- (X1);
        \draw[multiedge] (hs2) -- (X1);
        \draw[multiedge] (s2) -- (X2);
        \draw[multiedge] (hs2) -- (X2);

        \node[clique, label=-45:$W^{1,2}$] (W12) [right=0.25cm of Xk] {
        \begin{tikzpicture}
        \node[vertex, label=90:$w^{1,2}_{1,r_1}$] (1w11) at (0,0) {};
        \node[vertex, label=90:$w^{1,2}_{q,r_d}$] (1w1q) at (1,0) {};
        \draw (1w11) -- (1w1q);
        \node[fill=white, rectangle, rounded corners=false, inner sep=1pt] at (0.35,0) {\tiny $\dots$};
        \end{tikzpicture}
        };

        \node[clique, label=-45:$W^{1,3}$] (W13) [right=0.25cm of W12] {
        \begin{tikzpicture}
        \node[vertex, label=90:$w^{1,3}_{1,r_1}$] (2w11) at (0,0) {};
        \node[vertex, label=90:$w^{1,3}_{q,r_d}$] (2w1q) at (1,0) {};
        \draw (2w11) -- (2w1q);
        \node[fill=white, rectangle, rounded corners=false, inner sep=1pt] at (0.35,0) {\tiny $\dots$};
        \end{tikzpicture}
        };

        \node [right=0.1cm of W13] {$\dots$};

        \node[clique, label=-45:$W^{k-1,k}$] (Wk-1k) [right=0.75cm of W13] {
        \begin{tikzpicture}
        \node[vertex, label=90:$w^{k-1,k}_{1,r_1}$] (3w11) at (0,0) {};
        \node[vertex, label=90:$w^{k-1,k}_{q,r_d}$] (3w1q) at (1,0) {};
        \draw (3w11) -- (3w1q);
        \node[fill=white, rectangle, rounded corners=false, inner sep=1pt] at (0.35,0) {\tiny $\dots$};
        \end{tikzpicture}
        };

        \node[vertex, label=90:$\hat{d}^{1,2}$] (d12) [above=0.5cm of W12] {};
        \node[vertex, label=90:$\hat{d}^{1,3}$] (d13) [above=0.5cm of W13] {};
        \node[vertex, label=90:$\hat{d}^{k-1,k}$] (dk-1k) [above=0.5cm of Wk-1k] {};
        \node[vertex, label=90:$c^{1,2}$] (c12) [above right=0.25cm and 0.125cm of Xk] {};
        \node[vertex, label=90:$\hat{c}^{1,2}$] (hc12) [above right=1cm and 0.125cm of Xk] {};
        \node[vertex, label=90:$c^{1,3}$] (c13) [above right=0.25cm and 0.125cm of W12] {};
        \node[vertex, label=90:$\hat{c}^{1,3}$] (hc13) [above right=1cm and 0.125cm of W12] {};
        \node[vertex, label=0:$z$] (z) [below=0.5cm of Wk-1k] {};
        
        \draw[multiedge] (d12) -- (W12);
        \draw[multiedge] (d13) -- (W13);
        \draw[multiedge] (dk-1k) -- (Wk-1k);
        \draw[multiedge] (c12) -- (Xk);
        \draw[multiedge] (hc12) -- (Xk);
        \draw[multiedge] (c12) -- (W12.135);
        \draw[multiedge] (hc12) -- (W12.115);
        \draw[multiedge] (c13) -- (W12.45);
        \draw[multiedge] (hc13) -- (W12.65);
        \draw[multiedge] (c13) -- (W13.135);
        \draw[multiedge] (hc13) -- (W13.115);
        \draw[multiedge] (z) -- (Wk-1k);
        \draw (z) -- (hs1);
    \end{tikzpicture}
    }
    \caption{Construction of \cref{thm:w1-d2p}. If a vertex is adjacent to a box via a thick edge, then the vertex is adjacent to all vertices in that box. For \param{distance to path}, the ends of consecutive subpaths of the $X^i$ and $W^{i,j}$ are adjacent, for \param{distance to linear forest modules} they are not adjacent.}
    \label{fig:w1-d2p}
\end{figure}
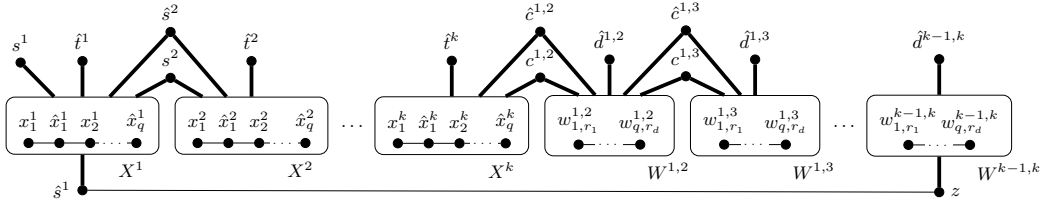

    We reduce the \textsc{Multicolored Clique Problem} to  POHPP parameterized by \param{distance to path}. The hardness of POHCP follows then by \cref{thm:path2cycle}. Let $G$ be an instance of the MCP. The vertex set of $G$ consists of $k$ color classes $V_1, \dots, V_k$ and every color class consists of vertices $v^i_1, \dots, v^i_q$.

    We construct a graph $G'$ as follows (see \cref{fig:w1-d2p}). The vertex set is partitioned in the following gadgets:
    \begin{description}
    \item[Selection Gadget] For every color class $i \in [k]$, we have a selection gadget consisting of a vertex $\hat{t}^i$ and a path $X_i$ that contains for every $v^i_p \in V_i$ a vertex $x^i_p$. These vertices are ordered by their index $p$ and after every vertex $x^i_p$ (and, thus, before $x^i_{p+1}$) there is a vertex $\hat{x}^i_p$ in that path. The vertex $\hat{t}^i$ is adjacent to all vertices in the path $X^i$.
    \item[Verification Gadget] For each pair $i,j \in [k]$ with $i < j$, we have a verification gadget consisting of a vertex $\hat{d}^{i,j}$ and a path $W^{i,j}$ containing for every edge $v^i_pv^j_r \in E(G)$ a vertex $w^{i,j}_{p,r}$. We do not fix an ordering of these vertices in that path. The vertex $\hat{d}^{i,j}$ is adjacent to all vertices in the path $W^{i,j}$.
    \end{description}

    Next we describe how these gadgets are connected to each other. We order the subpaths described in the selection and verification gadgets as follows: \[X^1, X^2, \dots, X^k, W^{1,2}, W^{1,3}, \dots, W^{1,k}, W^{2,3}, W^{2,4}, \dots, W^{k-1,k}.\]
    The last vertex of one of the paths is adjacent to the first vertex of the succeeding path.
    Therefore, these paths form one large path $\Psi$ in $G'$.
    Additionally, we have the following vertices in $G'$.
    First, we have $s^1$ and $\hat{s}^1$ that are adjacent to all vertices in $X^1$. For every $i \in [k]$ with $i > 1$, we have a vertex $s^i$ and a vertex $\hat{s}^i$ that both are adjacent to all vertices in $X^i$ and $X^{i-1}$. 
    For all $i,j \in [k]$ with $i < j$, there are vertices $c^{i,j}$ and $\hat{c}^{i,j}$ that are adjacent to all vertices in $W^{i,j}$ and to all vertices in the predecessor path of $W^{i,j}$ in the ordering described above.
    Finally, we have a vertex $z$ that is adjacent to all vertices in $W^{k-1,k}$ and to $\hat{s}^1$. Observe that the graph $G' - V(\Psi)$ contains $3k + 3\binom{k}{2} + 1$ vertices. Therefore, the \param{distance to path} of $G'$ is $\O(k^2)$.

    We define $Y$ as the set containing all the vertices defined above that have a hat in their name, i.e., $Y := \{\hat{s}^i \mid i \in [k]\} \cup \{\hat{t}^i \mid i \in [k-1]\} \cup \{\hat{c}^{i,j}, \hat{d}^{i,j} \mid i,j \in [k], i < j\}$. The partial order $\pi$ is the reflexive and transitive closure of the following constraints:
    \begin{enumerate}[(P1)]
        \item $s^1 \prec v$ for all $v \in V(G') \setminus \{s^1\}$,\label{p0}
        \smallskip
        \item $z \prec y$ for all $y \in Y$,\label{p6}
        \smallskip
        \item $x^i_p \prec w^{i,j}_{p,r}$ for all $i,j \in [k]$ with $i < j$ and all $p,r \in [q]$ with $v^i_pv^j_r \in E(G)$\label{p4}
        \smallskip
        \item $x^j_r \prec w^{i,j}_{p,r}$ for all $i,j \in [k]$ with $i < j$ and all $p,r \in [q]$ with $v^i_pv^j_r \in E(G)$\label{p5}
    \end{enumerate}

    \begin{claim}\label{claim:w1-d2p1}
        If $G$ has a multicolored clique $\{v^1_{p_1}, \dots, v^k_{p_k}\}$, then $G'$ has a $\pi$-extending Hamiltonian path $\cP$.
    \end{claim}

    \begin{claimproof}
        We start in $s^1$. Now we visit $x^1_{p_1}$ and go to $s^2$. Then we go to $x^2_{p_2}$. We repeat this process until we reach $x^k_{p_k}$. Next we visit $c^{1,2}$ and then $w^{1,2}_{p_1,p_2}$. Note that this is possible since both $x^1_{p_1}$ and $x^2_{p_2}$ are already visited and, thus, the constraints given in \pef{p4} and \pef{p5} are fulfilled. We then go to $c^{1,3}$ and visit $w^{1,3}_{p_1, p_3}$. We repeat this procedure until we reach $w^{k-1,k}_{p_{k-1}, p_k}$. Then we go to $z$ and then to $\hat{s}_1$. Now we traverse the path $X^1$. When we reach $\hat{x}^1_{p_1 - 1}$, we cannot visit the next vertex on the path as this vertex has already been visited. Therefore, we use $\hat{t}^1$ to jump over that vertex in $X^1$. When we have traversed $X^1$ completely, we visit $\hat{s}^2$ and then traverse $X^2$ in the same way as $X^1$. We repeat this procedure for all $X^i$ and also for all $W^{i,j}$, where we use $\hat{d}_{i,j}$ to jump over visited vertices.
    \end{claimproof}

    \begin{claim}\label{claim:w1-d2p2}
        If there is a $\pi$-extending Hamiltonian path $\cP$ in $G'$, then there is a multicolored clique $\{v^1_{p_1}, \dots v^k_{p_k}\}$ in $G$.
    \end{claim}

    \begin{claimproof}
        The path $\cP=(a_1,a_2,\dots a_n)$ has to start in $a_1=s^1$, due to \pef{p0}. Then \(a_2\) has to be some vertex $x^1_{p_1}$ for some $p_1 \in [q]$ since $s^1$ has no other neighbors. The next vertex (\(a_3\)) cannot be a neighbor of $x^1_{p_1}$ on the path $\Psi$, due to \pef{p6}. Hence, \(a_3\) has to be $s^2$. Its successor \(a_4\) either could be a vertex $x^2_{p_2}$ or a vertex $x^1_{p'}$. In the latter case, all the unvisited neighbors of $x^1_{p'}$ are not allowed to be taken next as they are forced to be to the right of $z$ by \pef{p6}. Therefore, $a_4 = x^2_{p_2}$ for some $p_2 \in [q]$. This argument can be repeated for every $i \in [k]$. Thus, the subpath of $\cP$ between $s^1$ and $c^{1,2}$ contains for any $i \in [k]$ the vertex $s^i$ and exactly one vertex $x^i_{p_i}$ for some $p_i \in [q]$. We claim that the vertices $C = \{v^1_{p_1}, \dots, v^k_{p_k}\}$ form a multicolored clique in $G$. 

        Using the same argument as above, the path $\cP$ has to go from $c^{1,2}$ to some vertex $w^{1,2}_{p,r}$. Due to \pef{p4}, $p$ must be equal to $p_1$. By \pef{p5}, $r$ has to be $p_2$. Therefore, $v^1_{p_1}$ and $v^2_{p_2}$ are adjacent. This observation also implies that the next vertex cannot be on $W^{1,2}$ but has to be $c^{1,3}$. Repeating this argumentation, we can show that $C$ in fact forms a multicolored clique in $G$.
    \end{claimproof}

    Obviously, the graph $G'$ can be constructed in $\O((kq)^2)$ time. Combining this with \cref{claim:w1-d2p1}, \cref{claim:w1-d2p2} and the fact that the \param{distance to path} of $G'$ is $\O(k^2)$, shows that $(G', \pi)$ is a valid \FPT{} reduction. This finalizes the proof for \param{distance to path}. 

    If we remove the edges between succeeding subpaths of $\Psi$ in $G'$, the graph induced by the vertices of $\Psi$ is a linear forest whose components are modules in $G'$. Therefore, the \param{distance to linear forest modules} of this graph is then $\O(k^2)$. Observe that the removed edges are not used in a Hamiltonian path in the proof above.  Therefore, the updated construction is a valid \FPT{} reduction for \param{distance to linear forest modules}.
\end{proof}

\subsection{Algorithms}

We start this section with an observation that follows from the fact that deleting a set of $k$~vertices from a traceable graphs produces at most $k+1$ components. This implies the following.

\begin{observation}
    Let $G$ be a graph with \param{vertex cover number}~$k$. If $G$ is traceable, then $G$ has at most $2k+1$ vertices. If $G$ is Hamiltonian, then $G$ has at most $2k$ vertices.
\end{observation}

This fact implies that traceable graphs with bounded \param{vertex cover number} have bounded values for all other graph width parameters (see \cref{fig:parameters}). Furthermore, the observation directly gives a kernel with a linear number of vertices and an \FPT{} algorithm for MinPOHPP and MinPOHCP.

We can give a similar result for graphs of bounded \param{treedepth}.

\begin{observation}\label{obs:treedepth}
    If a traceable graph has \param{treedepth} $k$, then $G$ has less than $2^k$ vertices.
\end{observation}

\begin{proof}
   \param{Treedepth} is monotone {\cite[Lemma~6.2]{nesetril2012sparsity}}, i.e, considering subgraphs does not increase the \param{treedepth}. Furthermore, the \param{treedepth} of a path with $n$ vertices is $\lceil \log_2(n+1) \rceil$  {\cite[Equation~6.2]{nesetril2012sparsity}}.
\end{proof}

Again, this observation implies that for traceable graphs with bounded \param{treedepth} all other graph width parameters are bounded (see \cref{fig:parameters}). As both MinPOHPP and MinPOHCP can be solved in exponential time~\cite{beisegel2024computing}, there are simple double-exponential \FPT{} algorithms for \param{treedepth}. For the \param{feedback edge set number} we can give a single-exponential \FPT{} algorithm. As was shown in the arXiv version of~\cite{demaine2019reconfiguring}, an $n$-vertex graph of \param{feedback edge set number}~$k$ has at most $2^k \cdot \binom{n}{2}$ different paths. Enumerating these paths leads to the following result.

\begin{theorem}\label{thm:fes}
     MinPOHPP and MinPOHCP can be solved in $2^k \cdot n^{\O(1)}$ time on an $n$-vertex graph of \param{feedback edge set number}~$k$.
\end{theorem}

Unless $\W = \FPT$, this result cannot be extended to \param{feedback vertex set number} since POHPP and POHCP are \W-hard for this parameter, due to \cref{thm:w1-d2p}. However, we can give an \XP{} algorithm for MinPOHPP when parameterized by \param{distance to outerplanar}, a more general parameter than \param{feedback vertex set number}. Before we present this algorithm, we first focus on a more restricted case that allows an \FPT{} algorithm. 

It has been shown by Beisegel et al.\ \cite{beisegel2024computing} that  MinPOHPP can be solved in $\O(n^2)$ time on outerplanar graphs. The key ingredient of that algorithm is the observation that the vertices of the outer face contained in a prefix of a Hamiltonian path must be an \emph{interval} of the outer face, i.e., a consecutive part of the closed walk that forms the face (see Lemma~5.2 in~\cite{beisegel2024computing}). 

\begin{lemma}[Beisegel et al.~\cite{beisegel2024computing}]\label{lemma:op-interval}
    Let $G$ be a planar graph and let $C$ be a face of a planar embedding of $G$. Then $V(C) \cap V(P)$ forms an interval of $C$.
\end{lemma}

Note that this result does not only hold for outerplanar graphs but for all plane graphs and all faces. It is important to note that this does not mean that the prefix contains this interval as a subpath. In fact, there can be vertices in between that are not part of the boundary of the face. This general result allows us to extend the polynomial-time algorithm for outerplanar graphs to an \FPT{} algorithm for plane graphs parameterized by the number of vertices that are not on the outer face. Note that our algorithm uses ideas similar to those of \cite{deineko2006traveling}.

\begin{theorem}\label{thm:inner-outerplanar}
    Given a plane graph $G$ with $k$ vertices not on the outer face, we can solve  MinPOHPP on $G$ in $\O(2^k k n^3)$ time and MinPOHCP on $G$ in $\O(2^k k n^5)$ time.
\end{theorem}
\begin{proof}
    We only sketch the idea of the algorithm as it is a quite straightforward adaption of Algorithm~2 in~\cite{beisegel2024computing}. For details, we refer to this publication. First note that considering induced subgraphs does not increase the number of inner vertices of an embedding. Hence, by \cite{beisegel2024computing}, it suffices to deal with 2-connected graphs.
    Let $C = (v_0, \dots, v_\ell)$ be the outer face of the planar embedding $P$ and $W$ be the set of vertices that are not on the outer face. Since $G$ is 2-connected, the face $C$ forms a cycle. Furthermore, let $c$ be the cost function on the edges of $G$.

    The idea of Algorithm~2 in~\cite{beisegel2024computing} is to use a dynamic programming approach. To this end, the authors consider tuples $(a,b,\omega)$, which represent prefixes of Hamiltonian paths. Here, $a$ and $b$ are the indices of the endpoints of the interval on $C$ associated with that prefix (see~\cref{lemma:op-interval}) and $\omega$ is either 1 or 2 depending on whether the prefix ends in $v_a$ or $v_b$. For the problem considered here, we simply need to adjust these to tuples of the form $(a,b,X,t)$, where $a$ and $b$ again represent the indices of the endpoints of the interval on $C$ associated with the prefix, $X \subseteq W$ is the set of inner vertices in the prefix and $t$ forms the endpoint of the prefix, respectively. There are at most $n^2$ intervals and at most $2^k$ subsets of $W$. The endpoint $t$ either has to be a vertex of $X$ or one of the vertices $v_a$ and $v_b$. Therefore, there are at most $(k+2)$ endpoints and the total number of tuples is bounded by $2^k (k+2) n^2$. Similar as in Algorithm~2 of \cite{beisegel2024computing}, we compute for every tuple $(a,b,X,t)$ the minimal cost $M(a,b,X,t)$ of an ordered path $\cP$ of $G$ fulfilling the following properties (or $\infty$ if no such path exists):
\begin{enumerate}[(i)]
  \item $\cP$ consists of the vertices in the interval between $v_a$ and $v_b$ of $C$ and in the set $X$,\label{cond:outer1}
  \item $\cP$ is a prefix of a linear extension of $\pi$,\label{cond:outer2}
  \item $\cP$ ends in $t$.\label{cond:outer3}
\end{enumerate} 
This is done inductively by the size of $\cP$. Let $a'$, $b'$, and $X'$ be the updated values if $t$ is removed from the potential prefix. We first have to check whether $t$ is minimal in $\pi$ if all the vertices of $[a',b']$ and $X'$ have been visited. This costs $\O(n)$ time. If this is not the case, we set the $M$-value to $\infty$. Otherwise, we check the $M$-values for all possible vertices $t'$ that are before $t$ in the prefix. Note that $t'$ can only be an element of $X'$ or one of the two vertices $v_{a'}$ and $v_{b'}$, due to \cref{lemma:op-interval}. For each choice of $t'$ we compute the value $M(a',b',X',t') + c(tt')$ and set the $M$-value of $(a,b,X,t)$ to the minimum of these values. Note that this can be done in $\O(k) \subseteq \O(n)$ time using an adjacency matrix.

Summing up, for each of the $\O(2^kkn^2)$ tuples we need $\O(n)$ time which leads to the overall running time of $\O(2^kkn^3)$. The proof of the correctness of the algorithm follows along the lines of the proof of Theorem~5.3 in \cite{beisegel2024computing}. The running time for MinPOHCP follows by \cref{thm:cycle2path}.
\end{proof}

The idea of that algorithm can also be used to give an \XP{} algorithm for \param{distance to outerplanar}. The crucial difference in the complexity is based on the fact that there is not only one interval of the outer face used by a prefix of a Hamiltonian path but at most $k+1$. To prove this, we need the following definitions. 

Let $C$ be a closed walk of a plane graph $G$. A subgraph $H$ lies \emph{inside} of $C$ if all vertices and edges of $H$ are elements of $C$ or are placed in a bounded region of $\mathbb{R}^2 \setminus C$. The subgraph $H$ lies \emph{strictly inside} of $C$ if it lies inside of $C$ and does not contain a vertex of $C$. A subgraph $H$ lies \emph{outside} of $C$ if all vertices and edges of $H$ are elements of $C$ or are placed in the unbounded region of $\mathbb{R}^2 \setminus C$.

\begin{lemma}\label{lemma:planar-k}
    Let $G$ be a graph with a set $W \subseteq V(G)$ of size $k$ such that $G - W$ is planar. Let $C$ be a face of a planar embedding of $G - W$ and let $V(C)$ be the vertices on the boundary of $C$. Let $P$ be a prefix of a Hamiltonian path of $G$. Then there exists a partition of $V(C) \cap V(P)$ into at most $k+1$ intervals of $C$.
\end{lemma}

\begin{figure}
    \centering
    \includegraphics{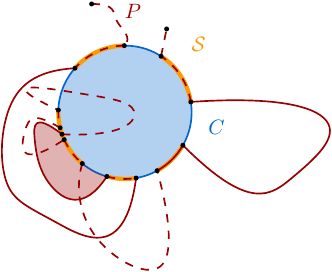}
    \caption{A face \(C\) (blue) and a prefix \(P\) of a Hamiltonian path (red). The solid parts of \(P\) are crests. The dashed parts might contain vertices of \(W\) and are thus not necessarily planar.}
    \label{fig:planar-k-definitions}
\end{figure}

\begin{proof}
    We use induction over the number of vertices of $G$. For one vertex the statement is trivial. So assume that the statement holds for all graphs with less than $n$ vertices. Let $G$ be a graph with $n$ vertices and a set $W \subseteq V(G)$ of size $k$ such that $G - W$ is planar. Let $C$ be a face of a planar embedding of $G - W$. Note that $C$ could be a closed walk with repeating vertices, for example, the outer face of a graph that is not 2-connected. In the following we will assume that this is not the case. This greatly simplifies the notation. By tweaking the definitions, the other case can be proven analogously, see the end of the proof.
    
    Let $P$ be a prefix of a Hamiltonian path $\cP$ of $G$. Let $\mathcal{S}$ be the smallest partition of $V(C) \cap V(P)$ into intervals of $C$, i.e., the intervals are inclusion-maximal. 
    See \cref{fig:planar-k-definitions} for an illustration. Let $|\mathcal{S}| = \ell$. If we remove the elements of $C$ from $P$, we get pairwise disjoint subpaths of $P$. We ignore the subpath that contains all vertices before the first vertex of $C$ on $P$ if it exists and the subpath containing all vertices after the last vertex of $C$ on $P$ if it exists. We also ignore all subpaths where the neighbor of the first and the last vertex lies on the same interval of $\mathcal{S}$. The endpoints of the remaining subpaths have neighbors on $P$ that lie on $C$. There are at least $\ell - 1$ of these subpaths. If each of these subpaths contains a vertex of $W$, then $\ell \leq k+1$ and we are done. Otherwise, consider a subpath $P_\text{sub}$ that does not contain vertices of $W$. We define the \emph{crest} induced by $P_\text{sub}$ as the subpath of $P$ containing $P_\text{sub}$ and the $P$-neighbors of its endpoints.
    Let $Q$ be such a crest. Let $v_i$ be the first vertex and $v_j$ be the last vertex of $Q$ with regard to the ordering of $P$. By definition, both $v_i$ and $v_j$ lie on $C$.
    We define $C[v_i,v_j]$ as the part of $C$ between $v_i$ and $v_j$ in the cyclic ordering and $C[v_j,v_i]$ the part of $C$ between $v_j$ and $v_i$ in the cyclic ordering. Further note that $v_i$ and $v_j$ lie in different intervals of $\mathcal{S}$. Therefore, there are vertices not in $P$, both in $C[v_i,v_j]$ and in $C[v_j,v_i]$. Let $t$ be the end vertex of $P$ or, in case this end vertex is $v_j$, let $t$ be the successor of $v_j$ in the Hamiltonian path $\cP$. Note that $t$ exists as $P$ does not contain all vertices of~$\cP$. W.l.o.g.~we may assume that $t$ lies outside of $\widetilde{Q} = Q \cup C[v_i,v_j]$. If this is not the case, then we define one face inside of $\widetilde{Q}$ as outer face without changing the faces of the embedding. Then, we can consider $Q \cup C[v_j,v_i]$ as $\widetilde{Q}$.
    
    Let $R$ be a \emph{minimal} crest inside of $\widetilde{Q}$, i.e., there is no crest that lies inside of $\widetilde{R} := R \cup C[x,y]$, where $x$ and $y$ are the endpoints of $R$ and $C[x,y]$ is the part of $C$ between $x$ and $y$ that is inside of $\widetilde{Q}$.
    Note that $C[x,y]$ contains vertices that are not part of $P$ as the crest $R$ connects two intervals, say $S_x$ and $S_y$, respectively.
    We will now show that we can always find either a vertex to delete, or an edge to contract in order to apply the induction hypothesis. 

Assume that there is an edge $ab \in E(P)$ such that w.l.o.g.~$b$ is strictly inside of $\widetilde{R}$. We contract the edge $ab$ in $G$ and $P$. The resulting graph $G'$ and resulting path $P'$ fulfill the induction hypothesis, i.e., $W \subseteq V(G')$, $P'$ is a prefix of a Hamiltonian path of $G'$ and $C$ is a face of an embedding of $G'$. Therefore, we can partition $V(C) \cap V(P')$ into at most $k+1$ intervals of $C$. These intervals map bijectively to intervals of $V(C) \cap V(P)$. Hence, we can assume that $P$ does not contain any vertices strictly inside of $\widetilde{R}$. 

   Suppose there is an interval $S$ that is contained completely in $C[x,y] - \{x,y\}$. As $t$ lies outside of $\widetilde{R}$ and $\widetilde{R}$ forms a cycle separating $S$ from $t$, the interval $S$ is adjacent to at least one vertex in $W$.
    If $S$ has exactly one $P$-neighbor $w \in W$ -- i.e. $P$ starts with $S$ -- then we delete $w$ and the whole interval \(S\) and connect the neighbors of $S$ on $C$ with an edge.
    This gives a graph \(G'\) with one less vertex and one less interval. The suffix of \(P\) that starts at the successor of \(P\) is the prefix of a Hamiltonian path in \(G'\). The statement follows by induction.
    
 Suppose $S$ has more than one $P$-neighbor in $W$. We connect these neighbors to form a clique. Furthermore, we sort them by their appearance in $P$ and contract the last and the penultimate vertex of these neighbors. As $t$ lies outside of $\widetilde{R}$, there are only vertices of $S$ visited in $P$ between these two neighbors. Furthermore, we delete \(S\) and connect its neighbors on $C$ with an edge. We can transform $P$ into $P'$ by replacing a sequence of $S$-vertices between two $P$-neighbors of $S$ by one of the clique edges. The statement follows by induction as there is one vertex less in $W$ and one interval less, see \cref{fig:planar-k-contract}.

 \begin{figure}
     \centering
     \includegraphics[page=3, scale=0.8]{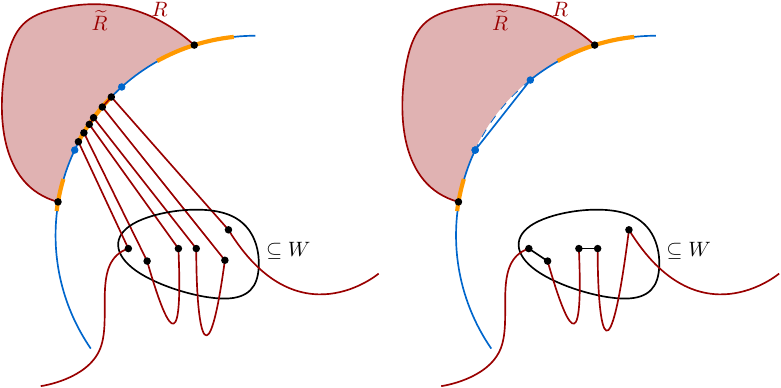}
     \caption{Illustration of \cref{lemma:planar-k}}.
     \label{fig:planar-k-contract}
 \end{figure}

    If there is no interval that is contained completely in $C[x,y] - \{x,y\}$, then all vertices on $C[x,y]$ are either in $S_x$, $S_y$ or not in $P$. In this case, we can apply similar arguments to contract the vertices that are not in $P$ and, thus, join $S_x$ and $S_y$ to reduce the number of intervals. Let $U$ be the set of vertices inside of $\widetilde{R}$ that are not in $P$. Note that the vertices in $U \cap C[x,y]$ are consecutive on $C$. As $t$ lies outside of $\widetilde{R}$, the vertices in $U$ have at least one neighbor in the complete Hamiltonian path $\cP$ that is contained in $W$. First suppose that the Hamiltonian path $\cP$ ends in some vertex of $U$. Let $w$ be the last vertex of $W$ on $\cP$. Note that all successors of $w$ in $\cP$ are in $U$. We delete $w$ and $U$ and connect the intervals $S_x$ and $S_y$ with an edge. Furthermore, we turn all remaining $\cP$-neighbors of $U$ in $W$ into a clique. The statement follows by induction as there is one vertex less in $W$ and one interval less.

    Suppose $\cP$ does not end in $U$. Then $U$ has at least two $\cP$-neighbors in $W$. We make these neighbors to a clique. Furthermore, we sort them by their appearance in $\cP$ and contract the last and the penultimate vertex of these neighbors. As $t$ lies outside of $\widetilde{R}$ and $\cP$ ends outside of $\widetilde{R}$, there are only vertices of $U$ visited in $\cP$ between these two neighbors. Furthermore, we delete $U$ and connect the intervals $S_x$ and $S_y$ with an edge. As above, the statement follows by induction as there is one vertex less in $W$ and one interval less.

    In the case that $C$ contains repeating vertices, the proof works similarly. However, we have to adapt the definition of what is ``inside'' of $\widetilde{Q}$ and $\widetilde{R}$. As $C$ contains repeating vertices, it might happen that $C[v_i,v_j]$ or $C[x,y]$ enclose vertices that are not inside of $\widetilde{Q}$ or $\widetilde{R}$. Such vertices can be handled in a similar way to those that are strictly inside of $\widetilde{Q}$ or $\widetilde{R}$.
\end{proof}

\begin{theorem}
    Given an $n$-vertex graph $G$ with \param{distance to outerplanar}~$k$, we can solve  MinPOHPP and MinPOHCP on $G$ in $n^{\O(k)}$ time.
\end{theorem}

\begin{proof}
    We can compute a set $W$ of size $k$ such that $G - W$ is outerplanar in $n^{\O(k)}$ time by brute-force as outerplanar graphs can be recognized in polynomial time~\cite{mitchell1979linear}. We adapt the algorithm given in \cref{thm:inner-outerplanar}. Let $C$ be the outer face of $G - W$. Note that $C$ could be a closed walk with repeating vertices. 

    We consider tuples $(\S,X,t)$, which again represent prefixes of Hamiltonian paths. Here, $\S$ is a set of segments of $C$, $X$ is a subset of $W$ and $t$ is the end vertex of the prefix. Due to \cref{lemma:planar-k}, we only need to consider sets of at most $k+1$ segments. We compute for every tuple $(\S,X,t)$ the minimal cost $M(\S,X,t)$ of an ordered path $\cP$ of $G$ fulfilling the following properties (or $\infty$ if no such path exists):
\begin{enumerate}[(i)]
  \item $\cP$ consists of the vertices in $\bigcup_{S \in \S} S$ and in the set $X$,\label{cond:outer21}
  \item $\cP$ is a prefix of a linear extension of $\pi$,\label{cond:outer22}
  \item $\cP$ ends in $t$.\label{cond:outer23}
\end{enumerate} 

This is done inductively by the size of $\cP$. Let $\S'$ and $X'$ be the updated values if $t$ is removed from the potential prefix. We first have to check whether $t$ is minimal in $\pi$ if all the vertices of $\S'$ and $X'$ have been visited. This costs $\O(n)$ time. If this is not the case, we set the $M$-value to $\infty$. Otherwise, we check the $M$-values for all possible vertices $t'$ that are before $t$ in the prefix. For each choice of $t'$ we compute the value $M(\S,X',t') + c(tt')$ and set the $M$-value of $(\S,X,t)$ to the minimum of these values.

To bound the running time, we first bound the number of tuples. There are at most $2^k$ possibilities for $X$, $n$ possibilities for $t$ and the number of sets $\S$ can be bounded by $n^{2(k+1)}$. Summing up, for each of the $\O(2^kn^{2k+3})$ tuples we need $n^{\O(1)}$ time which leads to the overall running time of $n^{\O(k)}$. The proof of the correctness of the algorithm follows along the lines of the proof of Theorem~5.3 in \cite{beisegel2024computing}.
\end{proof}

\section{Non-Sparse Width Parameters}\label{sec:d2c}

All parameters considered so far are sparse, i.e. unbounded for cliques. As POHPP and POHCP are trivial on cliques, it makes sense to consider also non-sparse parameters. Both MinPOHPP and MinPOHCP are \NP-hard on cliques as they form generalizations of (Path) TSP. Thus, we only consider the unweighted problems in this section.

\subsection{Hardness}

As mentioned in the introduction, \textsc{Hamiltonian Path} and \textsc{Hamiltonian Cycle} can be solved in \FPT{} time when parameterized by the \param{independence number}, that is the size of the largest independent set~\cite{fomin2025path}. For traceable graphs, this parameter is a lower bound on the \param{edge clique cover number}. Thus, both problems can also be solved in \FPT{} time for this parameter. Here, we show that these results cannot be extended to POHPP or POHCP unless $\P = \NP$.

\begin{theorem}\label{thm:ecc}
     POHPP and POHCP are \NP-complete on graphs of \param{edge clique cover number}~3.
\end{theorem}

\begin{proof}
    We reduce from the \NP-hard \textsc{Alternating Linear Extension Problem} for partial orders that are oriented from $A$ to $B$ to POHPP and POHCP. Let $(A,B,\pi)$ be an instance of this problem with $A = \{a_1, \dots, a_n\}$ and $B = \{b_1, \dots, b_n\}$ where $\pi$ is oriented from $A$ to $B$. We build a graph $G$ as follows (see \cref{fig:ecc} for an illustration). The vertex set of $G$ consists of the following vertices:
    \begin{itemize}
        \item $A = \{a_1, \dots, a_n\}$
        \item $B = \{a_1, \dots, a_n\}$
        \item $X = \{x_1, \dots, x_n\}$
        \item $Y = \{y_1, \dots, y_n\}$
        \item $Z = \{z_1, \dots, z_n\}$
    \end{itemize}

    The graph $G$ contains all edges except for those between vertices of $X$ and $B$, between vertices of $X$ and $Y$ as well as those between vertices of $Y$ and $A$. In other words, the edges of $G$ can be covered by the three cliques $C_1 = A \cup X \cup Z$, $C_2 = A \cup B \cup Z$, and $C_3 = B \cup Y \cup Z$.

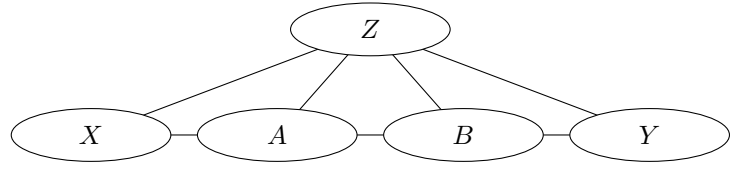
\begin{figure}
    \centering
    \begin{tikzpicture}[scale=0.7]
    \draw (0,0) -- (3.5,0);
    \draw (3.5,0) -- (7,0);
    \draw (7,0) -- (10.5,0);
    \draw (0,0) -- (5.25,2);
    \draw (3.5,0) -- (5.25,2);
    \draw (7,0) -- (5.25,2);
    \draw (10.5,0) -- (5.25,2);

    \draw[fill=white] (0,0) ellipse (1.5cm and 0.5cm);
    \draw[fill=white] (3.5,0) ellipse (1.5cm and 0.5cm);
    \draw[fill=white] (7,0) ellipse (1.5cm and 0.5cm);
    \draw[fill=white] (10.5,0) ellipse (1.5cm and 0.5cm);
    \draw[fill=white] (5.25,2) ellipse (1.5cm and 0.5cm);

    \node at (0,0) {$X$};
    \node at (3.5,0) {$A$};
    \node at (7,0) {$B$};
    \node at (10.5,0) {$Y$};
    \node at (5.25,2) {$Z$};
\end{tikzpicture}
    \caption{Construction of \cref{thm:ecc}. The ellipses form cliques. The connection of an ellipse with another implies that all edges between the vertices of the two cliques are present.}
    \label{fig:ecc}
\end{figure}
    
    The partial order $\pi'$ is the reflexive and transitive closure of the following constraints:

    \begin{enumerate}[(P1)]
        \item $x_1 \prec v$ for all $v \in V(G) \setminus \{x_1\}$,\label{ecc-p1}
        \item $x_1 \prec y_1 \prec z_1 \prec x_2 \prec y_2 \prec z_2 \dots \prec x_n \prec y_n \prec z_n$,\label{ecc-p2}
        \item $a_i \prec b_j$ if $(a_i,b_j) \in \pi$,\label{ecc-p3}
    \end{enumerate}

    \begin{claim}
        If $(A,B,\pi)$ has an alternating linear extension $\sigma$, then $G$ has a $\pip$-extending Hamiltonian path.
    \end{claim}

    \begin{claimproof}
        W.l.o.g.~we may assume that $\sigma = (a_1, b_1, a_2, b_2, \dots, a_n, b_n)$. Let $\cP$ be the following path (and cycle) in $G$:
        \[
        \cP := (x_1, a_1, b_1, y_1, z_1, x_2, a_2, b_2, y_2, z_2, \dots, x_n, a_n, b_n, y_n, z_n).
        \]
        We claim that $\cP$ is $\pip$-extending. The constraints \pef{ecc-p1} and \pef{ecc-p2} are obviously fulfilled. If there is a constraint $a_i \prec b_j$ in $\pi'$, then there is a constraint $a_i \prec b_j$ in $\pi$. Thus, $a_i$ is to the left of $b_j$ in $\sigma$ and, therefore, $a_i$ is to the left of $b_j$ in $\cP$. Hence, the constraints \pef{ecc-p3} are also fulfilled.
    \end{claimproof}

    \begin{claim}
        If $G$ has a $\pip$-extending Hamiltonian path, then $(A,B,\pi)$ has an alternating linear extension.
    \end{claim}
    Let $\cP$ be an $\pip$-extending Hamiltonian path in $G$. Due to constraint \pef{ecc-p2}, it holds that $x_i$ is to the left of $y_i$ in $\cP$. Furthermore, all vertices of $Z$ are either to the left of $x_i$ or to the right of $y_i$. This implies that between $x_i$ and $y_i$ there is at least one vertex of $A$ and one vertex of $B$ in $\cP$. As there are exactly $n$ vertices in $A$ and $n$ vertices in $B$, there must be exactly one vertex of $A$ and one vertex of $B$ between $x_i$ and $y_i$. W.l.o.g.~we may assume that for each $i \in [n]$, the vertices $a_i$ and $b_i$ are between $x_i$ and $y_i$. We claim that then $(a_1, b_1, \dots, a_n, b_n)$ forms an alternating linear extension of $\pi$. Assume for contradiction that there is a constraint $a_i \prec b_j$ with $j < i$. Then there is a constraint $a_i \prec b_j$ in $\pip$. However, this contradicts our assumption on $\cP$ as $b_j$ is to the left of $a_i$ in $\cP$.
\end{proof}

We have to leave open the case of \param{edge clique cover number}~2. However, we can observe that the graph $G$ constructed in the proof of \cref{thm:ecc} has \param{clique cover number}~2 since the cliques $C_1$ and $C_3$ cover all vertices. This implies the following result.

\begin{corollary}\label{thm:cc}
     POHPP and POHCP are \NP-complete on graphs of \param{clique cover number}~2.
\end{corollary}

Note that for graphs of \param{clique cover number}~1 (i.e., cliques), POHPP and POHCP are trivial as every linear extension of $\pi$ is a Hamiltonian path.

It follows from \cref{thm:w1-d2p} that POHPP and POHCP are \W-hard when parameterized by \param{distance to block}. Here, we extend this result to \param{distance to clique}. 

\begin{theorem}\label{thm:w1-d2c}
     POHPP and POHCP are \W-hard when parameterized by either \param{distance to clique} or \param{distance to cluster modules} also known as \param{twin-cover number}.
\end{theorem}

\begin{proof}
As for \cref{thm:w1-d2p}, we use a reduction from \textsc{Multicolored Clique} to POHPP and the use \cref{thm:path2cycle} to derive \W-hardness also for POHCP. Let $G$ be an instance of MCP with color classes $V_1, \dots, V_k$ where $V_i = \{v^i_1, \dots, v^i_q\}$. The construction works similar as for \cref{thm:w1-d2p}. The main difference is how we encode the selection of a vertex from a color class. While in the proof of \cref{thm:w1-d2p} the representative of the selected vertex was visited, here we visit all representatives except from the one of the selected vertex. 

We construct a graph $G'$ using the following gadgets:

    \begin{description}
    \item[Selection Gadget] For every color class $i \in [k]$, we have a clique $X^i$ that contains for every $v^i_p \in V_i$ a vertex $x^i_p$.
    \item[Verification Gadget] For each pair $i,j \in [k]$ with $i < j$, we have a clique $W^{i,j}$ containing for every edge $v^i_pv^j_r \in E(G)$ a vertex $w^{i,j}_{p,r}$.
    \end{description}

    Next we describe how these gadgets are connected to each other (see \cref{fig:w-clique}). 
    The union of the selection gadgets and the verification gadgets forms one large clique. We order the subcliques described in the selection and verification gadgets as follows: \[X^1, X^2, \dots, X^k, W^{1,2}, W^{1,3}, \dots, W^{1,k}, W^{2,3}, W^{2,4}, \dots, W^{k-1,k}.\] We have the following additional vertices in $G'$. First, we have $s^1$, $\hat{s}^1$, and $\hat{t}^1$ that are adjacent to all vertices in $X^1$. For every $i$ with $2 \leq i \leq k-1$, we have vertices $s^i$, $\hat{s}^i$ and $\hat{t}^i$. For $i = k$, we have only $s^i$ and $\hat{s}^i$. Vertex $s^i$ is adjacent to all vertices in $X^i$ and $X^{i-1}$ and the vertices $\hat{s}^i$ and $\hat{t}^i$ are adjacent to all vertices in $X^{i}$. Furthermore, $\hat{s}^i$ is adjacent to $\hat{t}^{i-1}$. For all $i,j \in [k]$ with $i < j$, there are vertices $c^{i,j}$ and $\hat{c}^{i,j}$ that are adjacent to all vertices in $W^{i,j}$ and to all vertices in the clique to left of $W^{i,j}$ in the ordering described above. Finally, we have vertices $z$ and $\hat{z}$ that are adjacent to all vertices in $W^{k-1,k}$. Furthermore, $z$ is adjacent to $\hat{s}^1$. Observe that the graph $G'$ without the selection and verification gadgets contains $3k + 2 \binom{k}{2} + 2$ vertices. Therefore, the \param{distance to clique} of $G'$ is $\O(k^2)$. 

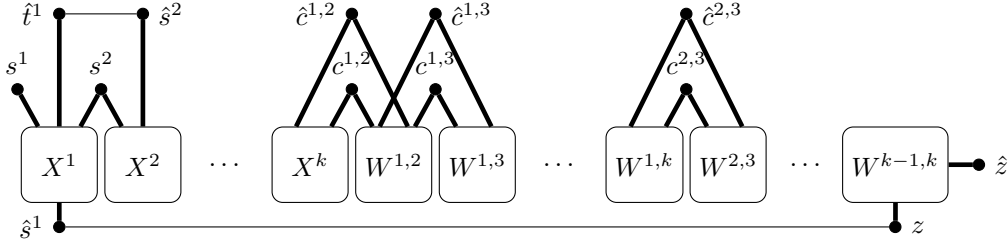
\begin{figure}
    \centering
    \begin{tikzpicture}[xscale=1.1, clique/.style={draw, rounded corners, minimum size=1cm}, multiedge/.style={line width=1.75}]
        \node[clique] (X1) at (0,0) {$X^1$};
        \node[clique] (X2) at (1,0) {$X^2$};
        \node at (2,0) {$\dots$};
        \node[clique] (Xk) at (3,0) {$X^k$};
        \node[clique] (W12) at (4,0) {$W^{1,2}$};
        \node[clique] (W13) at (5,0) {$W^{1,3}$};
        \node at (6,0) {$\dots$};
        \node[clique] (W1k) at (7,0) {$W^{1,k}$};
        \node[clique] (W23) at (8,0) {$W^{2,3}$};
        \node at (8.95,0) {$\dots$};
        \node[clique] (Wk-1k) at (10,0) {$W^{k-1,k}$};

        \node[vertex, label=90:{$s^1$}] (s1) at (-0.5,1) {};
        \node[vertex, label=90:{$s^2$}] (s2) at (0.5,1) {};
        \node[vertex, label=180:{$\hat{t}^1$}] (t2) at (0,2) {};
        \node[vertex, label=0:{$\hat{s}^2$}] (u2) at (1,2) {};
        \node[vertex, label=90:{$c^{1,2}$}] (c12) at (3.5,1) {};
        \node[vertex, label=180:{$\hat{c}^{1,2}$}] (d12) at (3.5,2) {};
        \node[vertex, label=90:{$c^{1,3}$}] (c13) at (4.5,1) {};
        \node[vertex, label=0:{$\hat{c}^{1,3}$}] (d13) at (4.5,2) {};
        \node[vertex, label=90:{$c^{2,3}$}] (c23) at (7.5,1) {};
        \node[vertex, label=0:{$\hat{c}^{2,3}$}] (d23) at (7.5,2) {};
        \node[vertex, label=0:{$z$}] (t1) at (10,-0.825) {};
        \node[vertex, label=0:{$\hat{z}$}] (z2) at (11,0) {};
        \node[vertex, label=180:{$\hat{s}^1$}] (u1) at (0,-0.825) {};

        \draw[multiedge] (s1) -- (X1);
        \draw[multiedge] (s2) -- (X1);
        \draw[multiedge] (s2) -- (X2);
        \draw[multiedge] (t2) -- (X1);
        \draw (t2) -- (u2);
        \draw[multiedge] (u2) -- (X2);
        \draw[multiedge] (d12) -- (Xk.110);
        \draw[multiedge] (c12) -- (Xk);
        \draw[multiedge] (c12) -- (W12);
        \draw[multiedge] (d12) -- (W12.70);
        \draw[multiedge] (d13) -- (W12.110);
        \draw[multiedge] (c13) -- (W12);
        \draw[multiedge] (c13) -- (W13);
        \draw[multiedge] (d13) -- (W13.70);
        \draw[multiedge] (d23) -- (W1k.110);
        \draw[multiedge] (c23) -- (W1k);
        \draw[multiedge] (c23) -- (W23);
        \draw[multiedge] (d23) -- (W23.70);

        \draw[multiedge] (t1) -- (Wk-1k);
        \draw[multiedge] (z2) -- (Wk-1k);
        \draw (t1) -- (u1);
        \draw[multiedge] (u1) -- (X1);
    \end{tikzpicture}
    \caption{Construction of \cref{thm:w1-d2c}. If a vertex is adjacent to a box via a thick edge, then the vertex is adjacent to all vertices in that box. For \param{distance to clique}, all the boxes are pairwise adjacent, for \param{distance to linear forest modules} they are not adjacent.}
    \label{fig:w-clique}
\end{figure}

    The partial order $\pi$ is the reflexive and transitive closure of the following constraints.
    \begin{enumerate}[(P1)]
        \item $s^1 \prec v$ for all $v \in V(G') \setminus \{s^1\}$,\label{d2c:p1}
        \item $v \prec \hat{z}$ for all $v \in V(G') \setminus \{\hat{z}\}$,\label{d2c:p1'}
        \item $s^1 \prec s^2 \prec \dots \prec s^k$,\label{d2c:p2}
        \item $z \prec \hat{s}^1 \prec \hat{t}^1 \prec \hat{s}^2 \prec \hat{t}^2 \prec \dots \prec \hat{s}^k$,\label{d2c:p3}
        \item $c^{i,j} \prec c^{i',j'}$ for all $i,j,i',j' \in [k]$ with $i < i'$ or $i = i'$ and $j < j'$,\label{d2c:p4}
        \item $c^{k-1,k} \prec z$, \label{d2c:p5}
        \item $c^{i,j} \prec w^{i,j}_{p,r}$ for all $i,j \in [k]$ with $i < j$ and all $w^{i,j}_{p,r} \in W^{i,j}$,\label{d2c:p6}
        \item $x^i_p \prec w^{i,j}_{r,s}$ for all $i,j \in [k]$ with $i < j$ and all $w^{i,j}_{r,s} \in W^{i,j}$ with $p \neq r$,\label{d2c:p7}
        \item $x^j_p \prec w^{i,j}_{r,s}$ for all $i,j \in [k]$ with $i < j$ and all $w^{i,j}_{r,s} \in W^{i,j}$ with $p \neq s$.\label{d2c:p8}
    \end{enumerate}

    \begin{claim}\label{claim:clique-direction1}
        If there is a multicolored clique $\{v^1_{p_1}, \dots, v^k_{p_k}\}$ in $G$, then there is a $\pi$-extending Hamiltonian path in $G'$.
    \end{claim}

    \begin{claimproof}
        First, we visit $s^1$ and then all the vertices in $X^1$ except for the vertex $x^1_{p_1}$. Then we visit $s^2$. We repeat this for all $i \in [k]$. 
        Next we visit $c^{1,2}$. Now we visit $w^{1,2}_{p_1, p_2}$. Note that this is possible since all the vertices that have to be to the left of $w^{1,2}_{p_1, p_2}$ by \pef{d2c:p7} and \pef{d2c:p8} are already visited. Next we visit $c^{1,3}$ and afterwards $w^{1,3}_{p_1, p_3}$ which is possible for the same reason as mentioned above. We repeat this procedure until we reach $w^{1,k}_{p_1,p_k}$. Next, we visit $c^{2,3}$. The same procedure works until we reach $w^{k-1,k}_{p_{k-1},p_k}$. Next we visit $z$. Now, starting with $i=1$, we visit for all $i \in [k]$ the vertices $\hat{s}^i$, followed by $x^i_{p_i}$ and $\hat{t}^i$ until we reach $x^k_{p_k}$. Then, we visit $\hat{c}^{1,2}$ followed by all remaining vertices $w^{1,2}_{p,r} \in W^{1,2}$ which is possible as all the left vertices in the constraints \pef{d2c:p7} and \pef{d2c:p8} have already been visited. We repeat this for all $i,j$ with $i<j$ in the lexicographic order and eventually we have have visited all vertices in $W^{k-1,k}$. Finally, we visit $\hat{z}$.
    \end{claimproof}

    It remains to show that a $\pi$-extending Hamiltonian path in $G'$ implies the existence of a multicolored clique in $G$. Assume that there is a $\pi$-extending Hamiltonian path $\cP$ in $G'$.

    \begin{claim}\label{claim:clique-selection1}
        For any $i \in [k]$, there is at least one vertex in $X^i$ that is to the right of $z$ in $\cP$.
    \end{claim}

    \begin{claimproof}
        By \pef{d2c:p3}, $\hat{s}^i$ is to the right of $z$. By \pef{d2c:p1'}, $\hat{s}^i$ is not the last vertex of $\cP$ and, thus, has two neighbors in $\cP$. One of these two neighbors has to be in $X^i$. This vertex is to the right of $z$.
    \end{claimproof}

    \begin{claim}\label{claim:clique-verification1}
        For any $i,j \in [k]$ with $i < j$, there is exactly one vertex of $W^{i,j}$ that is to the left of $z$ in $\cP$.
    \end{claim}

    \begin{claimproof}
        First assume for contradiction that there are two vertices $w^{i,j}_{p,r}$ and $w^{i,j}_{p',r'}$ to the left of $z$ in $\cP$. It holds that $p \neq p'$ or $r \neq r'$. We assume that $p \neq p'$, the case that $r \neq r'$ works analogously. Due to \pef{d2c:p7}, $x^i_p$ has to be to the left of $w^{i,j}_{p',r'}$ and $x^i_{p'}$ has to be to the left of $w^{i,j}_{p,r}$. All other vertices of $X^i$ have to be to the left of both $w^{i,j}_{p,r}$ and $w^{i,j}_{p',r'}$. Therefore, all vertices of $X^i$ are to the left of $z$ in $\cP$, contradicting \cref{claim:clique-selection1}.

        Now we show that there is also at least one vertex of $W^{i,j}$ to the left of $z$ in $\cP$.
        Consider vertex $c^{a,b}$ with $c^{a,b} \neq c^{1,2}$. Due to \pef{d2c:p1}, \pef{d2c:p4}, and \pef{d2c:p5}, vertex $c^{a,b}$ is not the start vertex of $\cP$ and $c^{a,b}$ is to the left of \(z\) in $\cP$. Thus, it has two neighbors in $\cP$. These neighbors have to be in the two $W$-cliques to which $c^{a,b}$ is adjacent. As we have observed above, there is at most one vertex of any of those cliques to the left of $z$. Therefore, both cliques contain a neighbor of $c^{a,b}$ in $\cP$ and these neighbors are to the left of $z$. As every clique $W^{i,j}$ is adjacent to a vertex $c^{a,b} \neq c^{1,2}$, it follows that every of those cliques contains a vertex that is to the left of $z$.
    \end{claimproof}

    \begin{claim}\label{claim:clique-selection2}
        For any $i \in [k]$, there is exactly one vertex $x^i_{p_i}$ that is to the right of $z$ in $\cP$.
    \end{claim}

    \begin{claimproof}
        Due to \cref{claim:clique-selection1}, it remains to show that there is at most one such vertex. As we have seen in \cref{claim:clique-verification1}, there is a vertex in $w^{i,i+1}_{p_i,p_{i+1}}$ or a vertex $w^{i-1,i}_{r_{i-1},r_i}$ to the left of $z$ in $\cP$. The constraints \pef{d2c:p7} or \pef{d2c:p8} imply that all vertices but $x^i_{p_i}$ or $x^i_{r_i}$, respectively, are to the left of~$z$.
    \end{claimproof}
       
    \begin{claim}\label{claim:clique-verification2}
        Let the $p_i$ be chosen as in \cref{claim:clique-selection2}. The set $\{v^1_{p_1}, \dots, v^k_{p_k}\}$ forms a multicolored clique in $G$.
    \end{claim}

    \begin{claimproof}
        Let $i,j \in [k]$ with $i < j$. Due to \cref{claim:clique-verification1}, there is a vertex $w^{i,j}_{a,b}$ to the left of $z$ in $\cP$. By \pef{d2c:p7} and \pef{d2c:p8}, all vertices $x^i_p$ and $x^j_r$ with $p \neq a$ and $r \neq b$ have to be to the left of $w^{i,j}_{a,b}$ in $\cP$. Therefore, \cref{claim:clique-selection2} implies that $a = p_i$ and $b = p_j$. By construction of $G'$, this implies that the edge $v^i_{p_i}v^j_{p_j}$ exists in $G$. 
    \end{claimproof}

    \Cref{claim:clique-direction1} and \cref{claim:clique-verification2} as well as the fact that the \param{distance to clique} of $G'$ is $\O(k^2)$ prove that $(G', \pi)$ is a proper \FPT{} reduction from MCP to POHPP parameterized by the \param{distance to clique}. This finalizes the proof for the case of \param{distance to clique}.

    Similar as in \cref{thm:w1-d2p}, we can adapt the construction to show that  POHPP is also \W-hard for \param{distance to cluster modules}. To this end, we delete all the edges in $G'$ between different gadgets. This works since these edges have not been used in the proof and since the resulting graph has \param{distance to cluster modules} $\O(k^2)$.
\end{proof}

\subsection{Algorithms}\label{sec:clique-algorithms}

Now we focus on POHPP and POHCP for the \param{distance to block}. First we observe that both problems can be solved in linear time if the graph is a block graph, i.e., its \param{distance to block} is~0. 

\begin{observation}\label{obs:block}
     POHPP and POHCP can be solved on block graphs in $\O(n + m + |\pi|)$ time, where $n$ is the number of vertices and $m$ is the number of edges of the given graph.
\end{observation}

\begin{proof}
    If the given graph is a clique, then any linear extension of $\pi$ is a solution of POHPP and POHCP. If the graph is not a clique, then it is not 2-connected and, therefore, it is not Hamiltonian. Thus, POHCP is trivial in that case. As was shown in \cite[Theorem~5.1]{beisegel2024computing}, a linear-time algorithm that solves POHPP on the blocks of a graph implies a linear-time algorithm that solves the problem on the whole graph. Therefore, we can solve  POHPP in $\O(n + m + |\pi|)$ time on block graphs.
\end{proof}

Next, we present an \FPT{} algorithm for \param{edge distance to block} when given a respective edge set.

\begin{theorem}\label{thm:block-edge}
     POHPP and POHCP can be solved in $(k+1)! \cdot 2^k \cdot n^{\O(1)}$ time on an $n$-vertex graph when given a set $F \subseteq E(G)$ with $|F| \leq k$ and $G - F$ is a block graph.
\end{theorem}
\begin{proof}
   Assume there is a $\pi$-extending Hamiltonian path of $G$. Then this path possibly uses some or all of the edges in $F$. We encode the way the path uses these edges by so-called edge choices. An \emph{edge choice} consists of a start vertex $u_0$ of the Hamiltonian path and an ordered subset $\hat{F}$ of $F$ such that for all edges in $\hat{F}$ one of the two directions in which the edge can be traverse is fixed. There are $\leq (k+1)!$ ordered subsets of $F$, at most $2^k$ choices of the directions and $n$ choices for the start vertex.
   
   Therefore, there are at most $(k+1)! \cdot 2^k \cdot n$ many edge choices.  In the following, we will call an edge choice \emph{valid} if there is a $\pi$-extending Hamiltonian path starting in $u_0$ such that all edges of $F$ that are part of this path are in $\hat{F}$ and are visited following the order that is fixed in that edge choice. 
   
   We first check whether the ordering of the vertices incident to $\hat{F}$ that is implied by our edge choice is a valid subordering of a linear extension of $\pi$. If this is not the case, we directly reject the edge choice. Otherwise, observe that vertices can be incident to more than one edge in an edge choice $\hat{F}$. However, in every valid edge choice, there are at most two edges of $\hat{F}$ incident to a particular vertex and these edges are consecutive in the ordering of $\hat{F}$. If this is the case, then we normalize the edge choice as follows. A sequence of edges in $\hat{F}$ where consecutive edges share an endpoint is replaced by a new edge going from the first vertex $u$ of this sequence to the last vertex $v$ of the sequence. The partial order is updated accordingly, i.e., all predecessors of some vertex in the sequence that are not part of the sequence become predecessors of $u$ and all successors that are not part of the sequence become successors of $v$.

    This normalization step produces a sequence of pairwise different vertices $(u_1, \dots, u_\ell)$ where for every \emph{even} number $i \in [\ell]$ vertices $u_{i-1}$ and $u_{i}$ are connected via an edge of $\hat{F}$ and $u_{i}$ and $u_{i+1}$ are not connected via an edge of $\hat{F}$. Note that it is possible that the start vertex $u_0$ is equal to $u_1$.
    
    Now, the task of the algorithm is to fill the remaining vertices into the gaps between the $u_i$'s or into the gap after $u_\ell$. Let $i \in [\ell] \cup \{0\}$ be even. Consider the block-cut tree $\T$ of $G - F$. For any path $P$ in $G - F$, we define the projection of $P$ into $\T$ as the path in $\T$ that contains for every vertex in $P$ either the block in which it lies if there is a unique one or the vertex itself if the vertex is a cut vertex. As $\T$ is a tree, any path between $u_i$ and $u_{i+1}$ in $G - F$ has the same projection into $\T$. We first compute these projections for every pair $(u_i, u_{i+1})$ where $i$ is even. If the path projection contains cut vertices, then we label these cut vertices with label $i$. Note that we can reject the edge choice if we have to relabel a cut vertex since a cut vertex cannot be used in two different projections. After this process we label all the unlabeled vertices -- in particular all the vertices that are not cut vertices -- with label $-1$. 

    We now traverse the vertices $u_i$ in increasing order. If $i$ is odd, then we directly go to $u_{i+1}$. Otherwise, we follow the projection between $u_i$ and $u_{i+1}$. Whenever we visit a block, we take all vertices that are possible due to $\pi$ in arbitrary order except for those vertices that have a label that is larger than $i$. Note that all such vertices are cut vertices. Further note that we might also visit cut vertices to blocks that we do not enter directly afterwards. Vertices that have label exactly $i$ are taken only if there is no other unvisited vertex in that block that can be taken. This is because these vertices are exactly those cut vertices that are used to enter the next block on the projection between $u_i$ and $u_{i+1}$ and, thus, we have to leave the block when we visit a vertex with label $i$.
    If we get stuck during this process, we reject that edge choice. Otherwise, we eventually reach vertex $u_\ell$. If the remaining vertices do not induce a connected graph, then we again reject the edge choice. Otherwise, we use the algorithm of \cref{obs:block} to check whether we can traverse these vertices starting in $u_\ell$. 
    
    It is obvious that a Hamiltonian path constructed by the algorithm is $\pi$-extending. It remains to show that the algorithm only fails to find a Hamiltonian path for some edge choice if there is no valid Hamiltonian path for that choice. To this end, assume that we reach a point where our algorithm gets stuck. This means there is no vertex minimal in the remainder of $\pi$ that is adjacent to the last visited vertex. In particular, either $u_{i+1}$ or the next cut vertex on the way to $u_{i+1}$ cannot be taken. Assume for contradiction that there is a $\pi$-extending Hamiltonian path $\cP$ in $G$ following our edge choice. Obviously, $\cP$ traverses the blocks of $G$ and the $u_i$'s in the same order as we have done. However, there must be a vertex $x$ that is traversed in $\cP$ during a block visit where our algorithm did not visit vertex $x$. Let $x$ be the leftmost vertex in $\cP$ that fulfills this property. There are two options why vertex $x$ was not visited by our algorithm. If its label did not fit, then $x$ would also not have been visited in this block in $\cP$  (since otherwise the path $\cP$ would not fit to our edge choice). The only other option is an unvisited vertex that is forced to be to the left of $x$ by $\pi$. However, this vertex would have been visited in $\cP$ contradicting the choice of $x$. Hence, $\cP$ cannot exist.
\end{proof}

Note that Dumas et al.~\cite{dumas2025polynomial} present a linear kernel for deciding whether the \param{edge distance to block} of a graph $G$ is at most $k$. However, they do not specify whether this kernel can be used to compute a respective set $F$ of $G$ in \FPT{} time.

We can use the algorithm given in \cref{thm:block-edge} to develop an \XP{} algorithm for the parameter \param{(vertex) distance to block}.

\begin{theorem}
    If the \param{distance to block} of a graph $G$ with $n$ vertices is $k$, then  POHPP and POHCP can be solved in $n^{\O(k)}$ time.
\end{theorem}

\begin{proof}
    First note that we can find a vertex set $W$ such that $G - W$ is a block graph and $|W| = k$ in time $n^{\O(k)}$ by enumerating all vertex subsets of size $k$. Fix one of these sets $W$. For every vertex in $W$, we choose one predecessor and one successor (or we decide that the vertex is the first or the last vertex of our Hamiltonian path). Furthermore, we fix an ordering of the vertices in $W$. We call these decisions a \emph{vertex choice}. There are $n^{\O(k)}$ many vertex choices. 
    
    For every of those vertex choices, we construct the graph $G'$ as follows. We delete all vertices of $W$ from $G$. Consider the pairs of predecessors and successors of the vertices in $W$. If they form a path starting and ending in vertices of $G - W$, we add an edge to $G'$ connecting the first vertex $s$ of this path with the last vertex $t$ of the path. We update the partial order $\pi$ accordingly, i.e., for any inner vertex on this path, we make its predecessors in $\pi$ to predecessors of $s$ and its successors in $\pi$ to successors of $t$. If we have chosen some vertex of $W$ as start vertex of the Hamiltonian path, then we add constraints to the partial order that makes $t$ (the first successor in $G-W$) the start vertex. Equivalently, if a vertex of $W$ is chosen to be the end vertex of the Hamiltonian path, then we make the vertex $s$ (the last predecessor in $G-W$) the end vertex of the partial order. We call the resulting updated partial order $\pip$.
    
   It is easy to see that a $\pi$-extending Hamiltonian path of $G$ that follows our vertex choice directly maps to a $\pip$-extending Hamiltonian path of $G'$ that traverses all the added edges in the order that is implied by the ordering of the vertices in $W$ and by the chosen predecessors and successors. Therefore, we can use the subroutine of \cref{thm:block-edge} that checks for the validity of an edge choice to decide the validity of our vertex choice here. This takes $n^{\O(1)}$ time. Thus, checking all the vertex choices needs $n^{\O(k)}$ time in total.
\end{proof}

As we have seen in \cref{thm:w1-d2c}, there is no \FPT{} algorithm for \param{distance to cluster modules}. However, we can give such an algorithm for \param{distance to clique module}.

\begin{theorem}
    Given an $n$-vertex graph $G$ with \param{distance to clique module}~$k$,  POHPP and POHCP can be solved in time $k! \cdot n^{\O(1)}$.
\end{theorem}
\begin{proof}
     First note that \param{distance to clique module} and a corresponding set $W$ can be computed in polynomial time since the largest clique module is equivalent to the largest set of vertices with the same closed neighborhood. Let $W$ be a set of $k$ vertices such that $C = G - W$ is a clique module of $G$. The algorithm considers all \(k!\) orderings $\rho = (v_1, \dots, v_k)$ of $W$ that fulfill the constraints of $\pi$. We call these orderings \emph{choices}.
    
    \proofsubparagraph*{Removing secluded vertices} We say that a vertex of $W$ is \emph{secluded} if it is not adjacent to the vertices in $C$. We observe that the predecessor and the successor of a secluded vertex in a Hamiltonian path has to be an element of $W$. Let $(v_i, v_{i+1}, \dots, v_j)$ be a subsequence of $\rho$ such that $v_i$ and $v_j$ are not secluded but $v_{i+1}, \dots, v_{j-1}$ are secluded. We call $v_i$ and $v_j$ \emph{frontier vertices}. 
    Note that secluded vertices are defined with regard to \(W\), while a vertex is a frontier vertex for a fixed vertex choice \(\rho\).
    If some consecutive vertices in that subsequence are not adjacent in $G$, then a Hamiltonian path cannot follow that ordering of $W$. Thus, we can reject that choice. Otherwise, we contract this sequence to an edge $v_iv_j$. If there is a constraint $x \prec_\pi v_q$ with $i \leq q \leq j$, then the vertex $x$ has to be visited before $v_i$. Hence, we add the constraint $x \prec v_i$ to $\pi$. Similarly, if there is a constraint $v_q \prec_\pi y$ with $i \leq q \leq j$, then $y$ has to be visited after $v_j$ and we add the constraint $v_j \prec y$ to $\pi$. Let $\sigma = (w_1, \dots, w_{k'})$ be the ordering of the subset of non-secluded vertices $W' \subseteq W$ that results from that normalization. We define $G' :=  G[W' \cup C]$ and we call the updated partial order $\pip$. The following claim is a direct consequence of the explanations above.

    \begin{claim}\label{claim:secluded}
        If there is a $\pi$-extending path of $G$ following the ordering $\rho$ of $W$, then there is a $\pip$-extending path of $G'$ following the ordering $\sigma$ of $W'$.
    \end{claim}

    \proofsubparagraph{Adding clique vertices to $\sigma$} We now describe how the algorithm constructs a Hamiltonian path following the ordering $\sigma$ of $W'$. To this end, we define for every vertex $x \in C$ the values $\ell(x) = \max (\{i \mid w_i \prec_\pip x\} \cup \{0\})$ and $r(x) = \min (\{i \mid x \prec_\pip w_i \} \cup \{k'+1\})$. That is $\ell(x)$ is the last \(w_i\in \sigma\) such that \(w_i\) has to be left of \(x\), while \(r(x)\) is the first \(w_i\in \sigma\) that has to be right of \(x\).

    \begin{claim}\label{claim:lr}
        Let $x,y \in C$. It holds that $\ell(x) < r(x)$ and $\ell(y) < r(y)$. If $x \prec_\pip y$, then $\ell(x) \leq \ell(y)$ and $r(x) \leq r(y)$.
    \end{claim}

    \begin{claimproof}
        Let $(a,b) = (\ell(x), r(x))$. Assume that $a \neq 0$ and $b \neq k'+1$ (otherwise, $\ell(x) < r(x)$ trivially holds ). Since $w_a \prec_\pip x \prec_\pip w_b$, it holds that $w_a \prec_\pip w_b$. As $\rho$ fulfills the constraints of $\pi$, $a < b$ is true. 
        
        Now assume that $x \prec_\pip y$ and let $(c,d) = (\ell(y), r(y))$. If $a=0$, then $a \leq c$ trivially holds. If $d=k' + 1$, then $b \leq d$ trivially holds. If $a \neq 0$, then we know that $w_a\prec_\pip x \prec_\pip y$. If $d \neq k' + 1$, then it holds $x \prec_\pip y \prec_\pip w_d$. By the definition of the functions $\ell$ and $r$, in both cases it holds that $a \leq c$ or $b \leq d$, respectively.
    \end{claimproof}

    For every pair $(i,j)$ with $0 \leq i < j \leq k' + 1$, we form a bucket $B_{i,j}$ containing those vertices $x$ with $(\ell(x), r(x)) = (i,j)$. The vertices in $B_{i,j}$ are ordered according to $\pip$, i.e., if there are two vertices $x,y \in B_{i,j}$ with $x \prec_\pip y$ then $x$ is to the left of $y$ in $B_{i,j}$. 

    The algorithm tries to fill the vertices of $C$ into the gaps between the $w_i$'s.
    We traverse $\sigma$ starting in $w_1$. Whenever we reach a vertex $w_i$, then we visit all the unvisited vertices $x \in C$ with $r(x) = i$ directly before $w_i$ following their ordering in $\pip$. Note that we can do this by visiting them in increasing order of $\ell(x)$, due to \cref{claim:lr}. We call these vertices \emph{forced vertices} since the partial order $\pip$ forces them to be to the left of $w_i$.  If $i \neq 1$ and $w_iw_{i-1} \notin E(G)$, then we have to visit some vertex of $C$ between $w_{i-1}$ and $w_i$. Hence, if no unvisited vertex with $r(x) = i$ exists, we have to choose another vertex for the gap. Let $(a,b)$ be the tuple with minimal $b$ such that $a < i$ and $B_{a,b}$ contains an unvisited vertex. Note that $b > i$ because all vertices with $r(x) \leq i$ have already been visited. If no such tuple $(a,b)$ exists, then we reject the choice $\rho$. Otherwise, we choose $x$ to be the leftmost unvisited vertex in $B_{a,b}$ and visit it between $w_{i-1}$ and $w_i$. We call $x$ an \emph{unforced vertex} since the partial order $\pip$ did not force it to be to the left of $w_i$. We repeat this process until we have visited $w_k$ or we have rejected the choice. Afterwards, we add all the unvisited vertices after $w_k$ following their ordering in $\pip$. We call the resulting ordering $\tau$.

     \begin{claim}\label{claim:taui}
        Let $\tau_i$ be the subordering that has been produced after $w_i$ was traversed. The ordering $\tau_i$ induces a path in $G'$ and forms a prefix of a linear extension of $\pip$.
    \end{claim}

    \begin{claimproof}
        First, we show that $\tau_i$ induces a path. If consecutive vertices are both in $C$, then they are adjacent by definition. If both of them are in $W'$, then they are adjacent since otherwise the algorithm would have added some vertex of $C$ between them. If one of them is in $C$ and the other is in $W'$, then they are also adjacent since there is no secluded vertex in $W'$. Hence, $\tau_i$ induces a path.

        It remains to show that $\tau_i$ is a prefix of a linear extension of $\pip$. The constraints on vertices of $W'$ are all fulfilled since we have checked this right at the beginning for $\rho$. If a vertex $x \in C$ was added left of some vertex $w_j \in  W'$, then $\ell(x) \leq j$ and, thus, $w_j \not \prec_\pip x$. If it has been added to the right of some vertex $w_j$, then $r(x) > j$ and, thus, $x \not\prec_\pip w_j$. Finally, consider the case that $x \prec_\pip y$ for some $x,y \in C$. Due to \cref{claim:lr}, it holds that $\ell(x) \leq \ell(y)$ and $r(x) \leq r(y)$. Consider the step in the algorithm where $y$ has been added to $\tau_i$. If $y$ is an forced vertex, then $x$ either was already visited (if $r(x) < r(y)$ or $x$ is an unforced vertex) or it has been added in the same step to the left of $y$. If $y$ is an unforced vertex, then $x$ has already been visited at this point since otherwise $x$ would have been chosen instead of $y$. Hence, $x$ is to the left of $y$ in $\tau_i$.
    \end{claimproof}

    This proves that the algorithm works correctly if it returns an ordering $\tau$. It remains to show that it also works correctly if it rejects the choice $\sigma$ of $W'$.

    \begin{claim}\label{claim:clique_module_no_path}
        If the algorithm rejects an ordering $\rho$ of $W$, then there is no $\pi$-extending Hamiltonian path of $G$ following that ordering of $W$.
    \end{claim}

    \begin{claimproof}
        Due to \cref{claim:secluded}, it is sufficient to show that there is no $\pip$-extending Hamiltonian path of $G'$ following $\sigma$. Assume for contradiction that such a Hamiltonian path exists. Let $w_i$ be the vertex where we have rejected and let $\tau_i$ be the ordering that was constructed till that step. We will show in the following that there is a $\pip$-extending Hamiltonian path following the decisions of our algorithm, contradicting the fact that our algorithm rejected $\rho$.
        
        First observe that the only option for a rejection is that we had to take an unforced vertex but we did not find a suitable one in the unvisited vertices. Let $\cP$ be a $\pip$-extending Hamiltonian path of $G'$ following $\sigma$ whose common prefix with $\tau_i$ is as long as possible. 
        Let \(x\) be the first vertex on \(\tau_i\) after the common prefix. 
        If there are multiple \(\pip\)-extending paths with this prefix, we choose one where \(x\) is leftmost.
        Note that this path is well-defined as $\tau_i$ cannot be a prefix of $\cP$ because otherwise the algorithm would have found a suitable unforced vertex. 
        Let $y$ be the vertex after the common prefix in $\cP$.

\begin{figure}
    \centering
    \includegraphics[page=2,width=\textwidth]{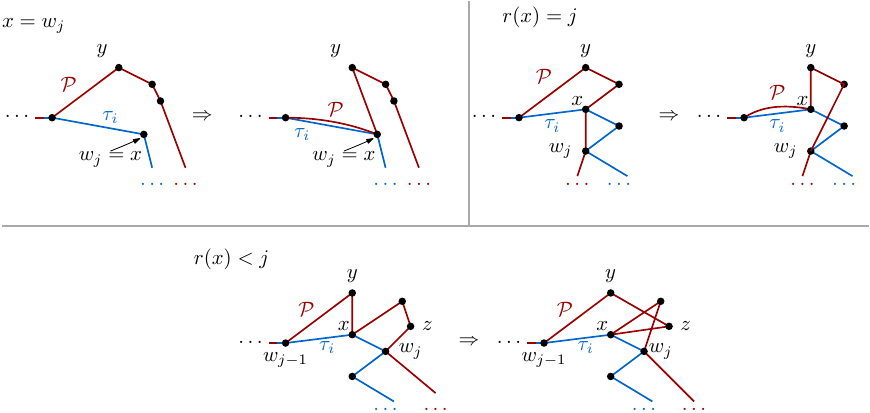}
    \caption{The three cases in the proof of \cref{claim:clique_module_no_path}}
    \label{fig:placeholder}
\end{figure}
        Now we consider the different cases for $x$. If $x$ is a vertex $w_j$, then $y$ and all other vertices between $y$ and $w_j$ in $\cP$ have to be unforced vertices, i.e., $r(z) > j$ for all these vertices $z$. We put all these vertices to the right of $w_j$ following their order in $\cP$. We claim that this results in another $\pip$-extending Hamiltonian path $\cP'$ of $G'$. It is obvious that the result is still a Hamiltonian path as the moved vertices are adjacent to all other vertices in $G'$. Furthermore, as observed above, none of the moved vertices is forced by $\pip$ to be to the left of $w_j$. Therefore, $\cP'$ is a $\pip$-extending Hamiltonian path with a longer common prefix with $\tau_i$; a contradiction to the choice of $\cP$.

        Next, assume $x$ is a forced vertex, i.e., $r(x) = j$ where $w_j$ is the leftmost vertex of $W'$ to the right of $x$ in $\tau_i$. We know that $x$ has to be between $y$ and $w_j$ in $\cP$. We just move $x$ to the front of $y$. Since $x$ and $y$ are not in $W'$, they form universal vertices in $G'$ and, thus, this results in another Hamiltonian path. By \cref{claim:taui}, the ordering $\tau_i$ is a prefix of a linear extension of $\pip$; hence this Hamiltonian path is also $\pip$-extending. Again this contradicts the choice of $\cP$ as the new Hamiltonian path has a longer common prefix with $\tau_i$.

        Finally, assume that $x$ is an unforced vertex, i.e., $r(x) > j$ where $w_j$ is the leftmost vertex of $W'$ to the right of $x$ in $\tau_i$. This implies that $x$ is the single vertex between the non-adjacent vertices $w_{j-1}$ and $w_j$ in $\tau_i$. This further implies that $y$ is also an unforced vertex since there must be a vertex of $C$ between $w_{j-1}$ and $w_j$ and all vertices $z$ with $r(z) = j$ are to the left of $x$ in $\tau_i$ (otherwise $x$ would not have been chosen by the algorithm). Due to the choice of the algorithm, we know that $r(y) \geq r(x)$. Let now $z \in C$ be the rightmost vertex in $\cP$ that is to the left of $x$ and fulfills the condition $r(z) \geq r(x)$. As observed above, \(y\) is such a vertex, and thus \(z\) exists. We swap $x$ and $z$ in $\cP$. It is obvious that the result is still a Hamiltonian path as $x$ and $z$ are universal vertices in $G'$. It remains to show that the path is still $\pip$-extending. Let $z'$ be a vertex between $z$ and $x$ in $\cP$. As vertex $z'$ is not part of the common prefix of $\tau_i$ and $\cP$, it is to the right of $x$ in $\tau_i$. Therefore, as $\tau_i$ is a prefix of a linear extension of $\pip$, the constraint $z' \prec x$ is not part of $\pip$. By the choice of $z$, it must hold that $r(z') < r(x) \leq r(z)$. Due to \cref{claim:lr}, the constraint $z \prec z' $ is also not part of $\pip$. Hence, the swap is allowed and results in another $\pip$-extending Hamiltonian path $\cP'$ of $G'$. However, either the common prefix of $\tau_i$ with $\cP'$ is longer than the common prefix with $\cP$ (if $y = z$) or $x$ is further left in \(\cP'\) than in $\cP$. In both cases, this contradicts the choice of $\cP$.
    \end{claimproof}

    \proofsubparagraph*{Adding back the secluded vertices}
    Assume that the algorithm has not rejected a choice $\rho$ and has produced an ordering $\tau$ of $G'$. 
    To get an ordering of $G$, we have to add the deleted secluded vertices again. To do this, we need the following observation.

    \begin{claim}\label{claim:frontier}
        No vertex of $C$ is put between two consecutive frontier vertices in $\sigma$.
    \end{claim}

    \begin{claimproof}
        As consecutive frontier vertices $w_i$ and $w_{i+1}$ are adjacent in $G'$, a vertex $x \in C$ would have been added between them only if $r(x) = i+1$. However, all vertices in $C$ that are to the left of $w_{i+1}$ in $\pip$ are also to the left of $w_i$ in $\pip$, due to the normalization step described above. Thus, $r(x) \neq i+1$ for all vertices $x \in C$.
    \end{claimproof}

    Due to this claim, we can add the removed secluded vertices back between their corresponding frontier vertices and get the ordering $\phi$ of $G$. The predecessors of secluded vertices in $\pi$ have been forced by $\pip$ to be to the left of the left frontier vertex. Similarly, their successors in $\pi$ were forced to be to the right of the right frontier vertex. Hence, the ordering $\phi$ induces a $\pi$-extending Hamiltonian path of $G$.

    As the algorithm considers at most $k!$ orderings of $W$ and for every of these orderings it needs polynomial time, the claimed running time holds.
\end{proof}

\section{Conclusion}

We have presented a comprehensive classification of the parameterized complexity of finding Hamiltonian paths or cycles with precedence constraints from the perspective of graph width parameters (see \cref{fig:parameters}). We were able to prove that all our \XP{} algorithms cannot be extended to \FPT{} algorithms (unless $\W = \FPT$). Nevertheless, our \W-hardness proofs cannot be used to show -- under the assumption of the Exponential Time Hypothesis -- that the algorithms' exponents are asymptotically optimal since these proofs increase the parameter value quadratically. Therefore, one may ask whether there are also \FPT-reductions where the parameter increases only linearly. 

So far, our work focused on the existence of parameterized algorithms. As a next step, one could also consider the question which of the \FPT{} results could be extended to kernels of polynomial size. As mentioned, there is a trivial polynomial kernel for \param{vertex cover number}. We are rather optimistic that such a kernel can also be given for \param{feedback edge set number}. 

One could also try to generalize \FPT{} or \XP{} algorithms to distance parameters for more general graph classes. One option are the \emph{series-parallel graphs}, a generalization of outerplanar graphs. These graphs are exactly the graphs of \param{treewidth}~2. Note that a generalization to \param{distance to treewidth~3} is hopeless as POHPP and POHCP are already para-\NP-hard for \param{edge distance to bandwidth~3}~\cite{beisegel2025graphi}. Another option could be considering graph width parameters not on all graphs but only on some subclass. In particular, the planar case of \param{distance to outerplanar} could be interesting as it lies between the general \W-hard case and the \FPT-case where the deleted vertices all lie in the interior of the outer face.

As POHPP and POHCP are already hard for very restricted graphs, one may ask how their complexity behaves if we restrict both the graph and the partial order. It has been shown \cite{beisegel2024computing} that POHPP is in \XP{} and \W-hard when parameterized by the partial order's \param{width}. Using similar arguments as for \cref{thm:path2cycle,thm:cycle2path}, we can show that the same holds for POHCP. However, it is open what happens if the problems are parameterized by both the partial order's \param{width} and some graph width parameter such as \param{treewidth} or \param{distance to clique}.

\bibliography{lit}
\bibliographystyle{plainurl}

\end{document}